\documentclass[a4size,hidelinks,journal]{IEEEtran}
\usepackage[english]{babel}
\usepackage{amsmath,amsfonts,amssymb, amsthm}
\usepackage{algorithmic}
\usepackage{cite} 
\usepackage{hyperref}

\newcommand{\algorithmicinitialize}{\textbf{Initialize:}}
\newcommand{\INITIALIZE}{\item[\algorithmicinitialize]}
\usepackage{algorithm}
\usepackage{array}
\usepackage[caption=true,font=footnotesize]{subfig}
\usepackage{textcomp}
\usepackage{stfloats}
\usepackage{url}
\usepackage{verbatim}
\usepackage{graphicx}
\def\BibTeX{{\rm B\kern-.05em{\sc i\kern-.025em b}\kern-.08em
T\kern-.1667em\lower.7ex\hbox{E}\kern-.125emX}}
\usepackage{balance}
\usepackage{orcidlink}
\usepackage{float}
\usepackage{lipsum}
\usepackage[acronym,shortcuts]{glossaries}
\theoremstyle{plain}
\newtheorem{theorem}{Theorem}
\newtheorem{lemma}[theorem]{Lemma}
\graphicspath{{./Fig_old/}}
\usepackage{placeins}
\usepackage{textcomp}
\usepackage{multirow}
\usepackage[most]{tcolorbox}

\newacronym{CF-mMIMO}{CF-mMIMO}{cell-free massive MIMO}
\newacronym{QoS}{QoS}{quality of service}
\newacronym{DSP}{DSP}{digital signal processing}
\newacronym{VLSI}{VLSI}{very large scale integration}
\newacronym{TX}{TX}{transmit}
\newacronym{RX}{RX}{receive}
\newacronym{IoT}{IoT}{Internet of Things}
\newacronym{SDR}{SDR}{semi-definite relaxation}
\newacronym{EVD}{EVD}{eigenvalue decomposition}
\newacronym{LDT}{LDT}{Lagrangian Dual Transform}
\newacronym{QT}{QT}{Quadratic Transform}
\newacronym{CDF}{CDF}{cumulative distribution function}
\newacronym{AP}{AP}{access point}
\newacronym{SINR}{SINR}{signal to interference-plus-noise ratio}
\newacronym{SIC}{SIC}{successive interference cancellation}
\newacronym{CSI}{CSI}{channel state information}
\newacronym{LoS}{LoS}{line-of-sight}
\newacronym{NLoS}{NLoS}{non-LoS}
\newacronym{MIMO}{MIMO}{multiple-input multiple-output}
\newacronym{MISO}{MISO}{multiple-input single-output}
\newacronym{SIMO}{SIMO}{single-input multiple-output}
\newacronym{SISO}{SISO}{single-input single-output}
\newacronym{MU}{MU}{multi-user}
\newacronym{JCAS}{JCAS}{joint communication and sensing}
\newacronym{JCR}{JCR}{joint communications and radar}
\newacronym{ISAC}{ISAC}{integrated sensing and communications}
\newacronym{3D}{3D}{three-dimensional}
\newacronym{2D}{2D}{two-dimensional}
\newacronym{1D}{1D}{one-dimensional}
\newacronym{ROI}{ROI}{region of interest}
\newacronym{mmWave}{mmWave}{millimeter-wave}
\newacronym{MF}{MF}{matched-filter}
\newacronym{SotA}{SotA}{state-of-the-art}
\newacronym{AWGN}{AWGN}{additive white Gaussian noise}
\newacronym{BS}{BS}{base station}
\newacronym{UE}{UE}{user equipment}
\newacronym{wlg}{w.l.g.}{without loss of generality}
\newacronym{CLT}{CLT}{central limit theorem}
\newacronym{PDF}{PDF}{probability density function}
\newacronym{ICI}{ICI}{inter-carrier interference}
\newacronym{BER}{BER}{bit error rate}
\newacronym{DoF}{DoF}{degrees-of-freedom}
\newacronym{VGA}{VGA}{vector Gaussian approximation}
\newacronym{FD}{FD}{full-duplex}
\newacronym{FP}{FP}{fractional programming}
\newacronym{CC}{CC}{communication-centric}
\newacronym{RC}{RC}{raised-cosine}
\newacronym{RRC}{RRC}{root raised-cosine}
\newacronym{6G}{6G}{sixth-generation}
\newacronym{V2X}{V2X}{vehicle-to-everything}
\newacronym{LEO}{LEO}{low-earth orbit}
\newacronym{I/O}{I/O}{input-output}
\newacronym{CE}{CE}{channel estimation}
\newacronym{ICC}{ICC}{integrated communication and computing}
\newacronym{ISCC}{ISCC}{integrated sensing, communications and computing}
\newacronym{PAM}{PAM}{pulse amplitude modulation}
\newacronym{iid}{i.i.d.}{independent and identically distributed}
\newacronym{MEC}{MEC}{mobile edge computing}
\newacronym{REMS}{REMS}{reconfigurable electromagnetic structure}
\newacronym{D-RIS}{D-RIS}{diagonal reconfigurable intelligent surface}
\newacronym{BD-RIS}{BD-RIS}{beyond-diagonal reconfigurable intelligent surface}
\newacronym{RIS}{RIS}{reconfigurable intelligent surface}
\newacronym{RE}{RE}{reflective element}
\newacronym{MRT}{MRT}{maximum ratio transmission}
\newacronym{ZF}{ZF}{zero forcing}
\newacronym{SVD}{SVD}{singular value decomposition}
\newacronym{CGA}{CGA}{conjugate gradient ascent}
\newacronym{QCQP}{QCQP}{quadratic constraint quadratic programming}
\newacronym{MMSE}{MMSE}{minimum mean square error}
\newacronym{RBD-RIS}{RBD-RIS}{reciprocal BD-RIS}
\newacronym{mMIMO}{mMIMO}{massive MIMO}
\newacronym{NRBD-RIS}{NRBD-RIS}{non-reciprocal BD-RIS}
\newacronym{CG}{CG}{conjugate gradient}
\newacronym{pp-ADMM}{pp-ADMM}{partially proximal alternating direction method of multipliers}
\newacronym{PDD}{PDD}{penalty dual decomposition}

\begin{document}
\title{Reciprocal Beyond-Diagonal Reconfigurable Intelligent Surface \!\!\! (BD-RIS):  \,Scattering\,Matrix\,Design\,via\,Manifold\,Optimization \!\!\!}



\author{
\IEEEauthorblockN{
~Marko Fidanovski\textsuperscript{\orcidlink{0009-0005-2926-1604}},~\IEEEmembership{Graduate Student Member,~IEEE},\\
~Iv\'{a}n Alexander Morales Sandoval\textsuperscript{\orcidlink{0000-0002-8601-5451}},~\IEEEmembership{Graduate Student Member,~IEEE},
~Hyeon Seok Rou\textsuperscript{\orcidlink{0000-0003-3483-7629}},~\IEEEmembership{Member,~IEEE}, ~Giuseppe Thadeu Freitas de Abreu\textsuperscript{\orcidlink{0000-0002-5018-8174}},~\IEEEmembership{Senior Member,~IEEE}, 
~Emil Björnson\textsuperscript{\orcidlink{0000-0002-5954-434X}},~\IEEEmembership{Fellow,~IEEE}.
}
\thanks{M. Fidanovski, I. A. M. Sandoval, H. S. Rou, G. T. F. de Abreu are with the School of Computer Science and Engineering, Constructor University (previously Jacobs University Bremen), Campus Ring 1, 28759 Bremen, Germany (emails: \{mfidanovski, imorales, hrou, gabreu\}@constructor.university).}
\thanks{E. Björnson is with the Department of Computer Science, KTH Royal Institute of Technology, Stockholm, Sweden (email: emilbjo@kth.se). He was supported by SweWIN and Digital Futures.}
\vspace{-3ex}
}

\markboth{To be submitted to the IEEE Transactions on Wireless Communications, 2025}%
{How to Use the IEEEtran \LaTeX \ Templates}

\maketitle

\begin{abstract}

\Acp{BD-RIS} are emerging as a transformative technology in wireless communications, enabling enhanced performance and \ac{QoS} of wireless systems in harsh urban environments due to their relatively low cost and advanced signal processing capabilities.
Generally, \ac{BD-RIS} systems are employed to improve robustness, increase achievable rates, and enhance energy efficiency of wireless systems in both direct and indirect ways.
The direct way is to produce a favorable propagation environment via the design of optimized scattering matrices, while the indirect way is to reap additional improvements via the design of \ac{MIMO} beamformers that further exploit the latter ``engineered'' medium.
In this article, the problem of sum-rate maximization via \ac{BD-RIS} is examined, with a focus on feasibility, namely low-complexity physical implementation, by enforcing reciprocity in the \ac{BD-RIS} design in a manner that adheres to the geometry of the manifold of symmetric matrices.
To that end, the sum-rate objective is transformed into a quadratic function via \ac{FP}, augmented via the also quadratic reciprocity constraint in the form of a regularization term, while the unitary constraint is dealt with via a manifold optimization framework.
Simulation results demonstrate the effectiveness of the proposed method in outperforming current \ac{SotA} approaches in terms of sum-rate maximization.
\end{abstract}


\glsresetall

\section{Introduction}
\IEEEPARstart{I}{ntelligent} metasurfaces, or \acp{RIS}, are an emerging resource in the area of wireless communications, offering a new paradigm for controlling foundational features of electromagnetic fields to enhance system performance \cite{BjornsonSPM2022, DiRenzoJSAC2020, LiCM2024, BasarVTM2024, ElMossallamyTCCN2020, LiArX2025}.
These surfaces consist of individual elements, $i.e.$, passive circuits, with controllable properties that allow them to filter and reradiate the incoming signals without amplification, which is accomplished by implementing delays, reducing the amplitude, or changing wave polarization via the circuit's local reprogrammable impedance.

Thanks to their low-power and reconfigurable nature, the so-called wave-domain signal processing approach offered by \acp{RIS} offer several key advantages over classical analog/digital signal processing, including ease of deployment and substantial reductions in end-to-end latency and processing delay \cite{LiuCST2021}.
%

Motivated by these advantages, \acp{RIS} have attracted enormous attention as an effective solution for enhancing wireless channel quality and consequently the communication performance of \ac{MISO}, \ac{MIMO}, and \ac{ISAC} systems \cite{BjornsonSPM2022}.
Indeed, intelligent surfaces revolutionize wireless systems design by allowing dynamic reconfiguring of the propagation environment, achieving substantial gains without the need for additional active radio chains \cite{NeriniTWC2023, RanasingheArx2025_FIM
}.

Early research on the approach predominantly focused on \ac{D-RIS} architectures, whose scattering matrices enabled simpler analysis and straightforward implementation \cite{ShtaiwiTC2023 ,YuICCC2019, HuangTWC2019,WuTWC2019,
RexhepiArX2025, GuoTWC2020, LiArX2023 
}.
Their use has been extensively studied across various communication settings under various design objectives, including sum-rate maximization 
\cite{ShtaiwiTC2023},  
spectral efficiency improvement \cite{YuICCC2019}, energy efficiency enhancement \cite{HuangTWC2019}, transmit power reduction \cite{WuTWC2019}, and physical layer security \cite{RexhepiArX2025}.
%
To mention a few illustrative examples, in \cite{GuoTWC2020} a weighted sum-rate maximization problem under perfect and imperfect \ac{CSI} was proposed, with algorithms developed for each case. 
Alternatively, in \cite{LiArX2023}, an active \ac{RIS}-aided downlink rate-splitting multiple access system is proposed, and its sum-rate performance was analyzed via a novel algorithm.
To this extent, a joint optimization problem was formulated to design the \ac{TX} beamforming and rate allocation vectors together with the \ac{RIS} scattering matrix.
The approach involves a two-stage design, where \ac{FP} techniques were used in combination with \ac{QCQP} methods to optimize the respective variables.
Another excellent example is \cite{RexhepiArX2025}, whereby the authors utilize the \ac{D-RIS} in an \ac{ISAC} framework, applying manifold optimization to jointly improve sum-rate performance and enhance security against eavesdropping.

Although the conventional diagonal architecture offers analytical convenience, it also imposes structural constraints that limit the \ac{DoF}, restricting its ability to manipulate the signal space.
To overcome these limitations, subsequent studies have focused on \ac{BD-RIS} architectures, which significantly increase the \ac{DoF} available for signal shaping \cite{JuniorJCS2024, KhanPIIC2023, LiTWC2023, YahyaOJCS2024, FangCL2024, ZhouTWC2024, 
LiuArx2025, LiArx2025_NRBD-RIS, 
Demir2024Wideband}. 
However, the increase in \ac{DoF}, achieved through enhanced connectivity of the \ac{BD-RIS}, is generally accompanied by higher hardware complexity and implementation cost. 
As a result, a substantial body of literature has investigated the trade-off between achievable performance and hardware complexity in \ac{BD-RIS} architectures \cite{Nerini2023ParetoBDRIS, Wu2025BeyondDiagonalRIS}.

In this context, advanced optimization techniques have been employed to effectively exploit the additional \ac{DoF} while respecting structural constraints. 
Both \ac{FP} and first-order methods, including projected gradient and accelerated gradient schemes \cite{YangTWC2024}, have been widely adopted in \ac{RIS} optimization problems to improve convergence speed and computational efficiency. 
While these approaches typically enforce constraints through Euclidean projections onto feasible sets, manifold optimization instead leverages the intrinsic geometry of the constraint space. 
By operating directly on smooth manifolds via tangent-space projections and retraction steps \cite{AbsilPSP2008,boumal2023introduction}, it maintains feasibility along the optimization trajectory while explicitly accounting for the geometry of the constraint set.
Junior \emph{et al.} \cite{JuniorJCS2024} provide a comprehensive tutorial on manifold optimization techniques and their application in \ac{RIS}-aided \ac{mMIMO} systems, where it is shown that adhering to the manifold geometry of the constraints consistently yields better solutions than standard optimization methods, which motivates the approach adopted in this article.
Several recent works have applied these principles to specific \ac{BD-RIS} design problems, employing manifold optimization to enhance secrecy rate \cite{KhanPIIC2023} and to jointly configure the scattering matrix and \ac{TX} precoder \cite{LiTWC2023}.
Both these as well as similar studies such as \cite{LiTVT2023,LiTWC2024}, however, are conducted for non-reciprocal settings, without enforcing symmetry.

Enforcing a symmetric scattering matrix ensures compatibility with physically realizable passive architectures, avoiding the need for active components, magnetic biasing, or non-reciprocal circuit elements, which increase hardware complexity, power consumption, and implementation cost \cite{LiArX2025, LiuArx2025}.
One way to enforce the unitary and symmetry constraints inherent to reciprocal BD-RIS is through a \ac{PDD} framework, as proposed in \cite{ZhouTWC2024}, where the scattering matrix and transmit beamformers are jointly optimized for both sum-rate maximization and transmit power minimization problems.
While effective, the approach relies on auxiliary-variable decoupling and a double-loop procedure, is primarily tailored to fully- and group-connected BD-RIS architectures, and incurs relatively high computational complexity. 
To address these limitations, \cite{ZheyuArxiv2024} develops an architecture-independent framework based on \ac{pp-ADMM}, which reformulates the BD-RIS constraints via the admittance matrix and enables efficient closed-form or low-complexity updates at each iteration. 
This general and computationally efficient framework is adopted in this work as a benchmark for performance and complexity comparisons.

%
%

%
%

Alternative \ac{BD-RIS} design strategies extend beyond direct optimization from the sum-rate objective and instead focus on hybrid approaches.
For example, in \cite{YahyaOJCS2024}, two methods for designing \ac{BD-RIS} scattering matrices are proposed.
The first method seeks to maximize the total signal power received by the users following a passive \ac{MRT} strategy, while the second approach aims at nullifying the multi-user interference via a passive \ac{ZF} strategy. 
In contrast, this work adopts a fundamentally different strategy: rather than decoupling the design into separate \ac{MRT} or \ac{ZF} criteria, the scattering matrix is optimized with respect to the sum-rate objective.
This direct formulation accounts for both signal enhancement and interference suppression, enabling these aspects to be addressed simultaneously within a unified and coherent optimization framework.

Another notable example is presented by Fang \emph{et al.} \cite{FangCL2024}, where the authors propose an alternative methodology---leveraging convex relaxation techniques---to address the non-convex \ac{BD-RIS} problem.
Specifically, they reformulate the non-convex constraints into a convex set, enabling efficient solution of the relaxed problem.
The solution is then mapped back to the feasible region of the original non-convex set using their proposed symmetric unitary projection.

More recently, research has shifted toward \ac{NRBD-RIS} designs, particularly within a full-duplex framework \cite{LiuArx2025,LiArx2025_NRBD-RIS}.
For example, Liu \emph{et al.} \cite{LiuArx2025} investigate their application in wireless circulator systems, demonstrating that properly optimized non-reciprocal scattering matrices yield improvements in both spectral efficiency and interference management compared to existing reciprocal structures, further expanding the applicability of \ac{BD-RIS} in wireless networks.
However, the practical realization of non-reciprocal \ac{BD-RIS} entails increased hardware complexity and resource requirements, which may limit their suitability in large-scale or cost-sensitive deployments.
Although non-reciprocal designs can be advantageous in specific scenarios, reciprocal \ac{BD-RIS} architectures remain highly relevant for the broader range of applications where simplicity, scalability, and energy efficiency are essential.

All in all, we recognize the maximization of sum-rate via the design of reciprocal \ac{BD-RIS} scattering matrices as a \emph{dual manifold} problem which, to the best of our knowledge, has not yet been addressed as such via an adequate manifold optimization method, likely due to the challenge of dealing with two manifolds of different geometries simultaneously.
Recently, this challenge was addressed in \cite{santamaria2026BDRIS}, where a Riemannian optimization algorithm operating on the solution space defined by the intersection manifold of unitary and symmetric matrices was proposed.\footnote{We emphasize that a pre-print of \cite{santamaria2026BDRIS} first appeared only in Jan. 2026 \url{https://arxiv.org/abs/2601.13877}, after the submission of this article which was available at \url{http://www.arxiv.org/abs/2509.20246} since Sep. 2025.} 
That formulation is, however, restricted to the fully-connected architecture and primarily considers the single-user case, where inter-user interference is absent.

In contrast to the latter, in this article, the dual-manifold nature of the reciprocal \ac{BD-RIS} scattering matrix design problem is addressed under a general setting\footnote{For the sake of simplicity and to allow direct comparison with \ac{SotA} methods we explicitly focus on \ac{MU}-\ac{MISO} systems, for which closed-form gradients are derived, while emphasizing that the proposed method applies directly to the general \ac{MU}-\ac{MIMO} case. See Subsection \ref{sec:MUMIMO_Generalization} for more details.} by incorporating the reciprocity constraint as a quadratic regulation term into the sum-rate objective transformed also into a quadratic function via \ac{FP}, while enforcing the unitarity constraint by solving the resulting regularized problem via manifold optimization over the Stiefel manifold.
Our main contributions are summarized as follows:
\begin{itemize}
\item {\bf Novel technique with improved performance:} A novel method is proposed to design the scattering matrix of \ac{RBD-RIS}, under which both the symmetry and unitary constraints are sensibly enforced, in a manner that accommodates both manifold geometries.
In particular, in the proposed method, the quadratic geometry of the symmetry constraint manifold is incorporated into the quadratic sum-rate objective reformulated via \ac{FP}, while the geometry of the Stiefel unitary constraint manifold is dealt with via manifold optimization.
The new method, which has complexity similar to that of \ac{SotA} techniques, is shown via simulations to outperform the latter under equivalent conditions.
\item {\bf Original closed-form expressions:} While the method is applicable to general \ac{MU}-\ac{MIMO} systems, a closed-form expression for the gradient of the proposed regularized and \ac{FP}-reformulated objective function is provided for the \ac{MU}-\ac{MISO} setup, enabling the implementation of the method via \ac{CGA} in that case. 
\item {\bf High hardware compatibility:} Thanks to the above, the proposed scattering matrix design of \ac{RBD-RIS} under \ac{MU}-\ac{MISO} settings can be implemented in a highly controlled manner that facilitates efficient hardware design \cite{parhi1999vlsi}.
In particular, with the exception of an Armijo line search, the key steps of the method, namely, the tangential projection, gradient evaluation, and updating of \ac{FP} variables,  are achieved via successive evaluation of closed-form expressions, enabling dedicated signal processing hardware to be optimized.\footnote{It is a well-known and widely recognized principle in \ac{DSP} and \ac{VLSI} designs that closed-form algorithms reduce to a fixed sequence of arithmetic operations, favoring the stability, reliability, power consumption and efficacy of hardware implementation, as a result of the corresponding deterministic pipelines, parallelization and static scheduling \cite{parhi1999vlsi}.}
\end{itemize}

\textit{Organization:} Section \ref{sec:sysmodel} introduces the \ac{BD-RIS}-aided \ac{MU}-\ac{MISO} system model, presents the different architectures considered, and formulates the problem addressed in the article. 
Next, a general solution to the first stage of the optimization problem, namely the scattering matrix design, is provided in Section~\ref{sec:smdesign}.
Simulation results are offered in Section~\ref{sec:simresults}, which help evaluate the performance of the proposed designs.
Finally, a brief conclusion is given in Section~\ref{sec:conc}, followed by Appendix \ref{app:lemma1} and \ref{app:lemma2}, where proofs of two Lemmas are detailed.

\textit{Notation:} Unless otherwise specified, $\mathbf{X}$ and $\mathbf{x}$ denote matrices and vectors. 
The absolute value, $\ell^2$ and Frobenius norm are denoted by $|\cdot|$,  $\|\cdot\|_2$, $\|\cdot\|_F$. 
The transpose and the Hermitian transpose are denoted by $(\cdot)^T$ and $(\cdot)^H$.
In addition, $\mathbf{X}_{i,j}$ denotes the $i$-th row and $j$-th column element of the matrix $\mathbf{X}$, $\mathbf{X}_{i:\bar{i},j:\bar{j}}$ extracts the elements from the $i$-th to the $\bar{i}$-th row, and $j$-th to $\bar{j}$ column from the matrix $\mathbf{X}$.
Furthermore, $\text{diag}(\mathbf{x})=(x_1,x_2,\dots,x_k)$ denotes a diagonal matrix, where the main diagonal is $\mathbf{x}$, and $\text{blkdiag}\left(\mathbf{X}_1, \mathbf{X}_2, \dots, \mathbf{X}_k\right)$ denotes a $G\times G$ block diagonal matrix with off diagonal elements $0$, where the blocks are $\mathbf{X}_1, \mathbf{X}_2, \dots, \mathbf{X}_k$.
The sets of complex and real numbers are denoted by $\mathbb{C}$ and $\mathbb{R}$, while the real and imaginary values of a complex number are denoted by $\Re\{\cdot\}$ and $\Im\{\cdot\}$, respectively.
The circular-symmetric complex normal random variable with $\sigma^2$ variance is denoted as $\mathcal{CN}(0, \sigma^2)$.

\section{System Model}
\label{sec:sysmodel}


Consider a \ac{BD-RIS} aided downlink \ac{MU}-\ac{MISO} system as illustrated in Figure~\ref{fig:system model}, where a \ac{BS} with $N$ \ac{TX} antennas serves $K$ single \ac{RX} antenna users\footnote{This scenario is considered due to its analytical tractability and simplicity, while enabling direct comparisons with \ac{SotA} alternatives. The methods can, however, be easily generalized to \ac{MU}-\ac{MIMO} settings, as described in Subsection \ref{sec:MUMIMO_Generalization}.}
with the help of a \ac{BD-RIS} consisting of $R$ \acp{RE}, 
The {associated} \ac{BD-RIS} scattering matrix is represented by $\mathbf{\Theta} \in \mathbb{C}^{R \times R}$.
The channel linking the \ac{BS} and the \ac{BD-RIS} ($i.e.$, the \ac{BS}-{to}-\ac{BD-RIS} channel)  is denoted as $\mathbf{H}_{\mathrm{TX}} \in \mathbb{C}^{R \times N}$, and the {channel vector connecting the \ac{BD-RIS} and the $k$-th user} is $\mathbf{h}_k \in \mathbb{C}^{R \times 1}$.

The transmit signal vector is given by $\mathbf{x}=\mathbf{V}\mathbf{s}$, where the information symbols $\mathbf{s} \in \mathbb{C}^{K \times 1}$ satisfy $\mathbb{E}[\mathbf{s}\mathbf{s}^H]=\mathbf{I}$, and the beamforming matrix $\mathbf{V} \in \mathbb{C}^{N \times K}$ meets the power constraint $\|\mathbf{V}\|_F^2 \leq P_{\mathrm{max}}$, with $P_{\mathrm{max}}$ being the maximum transmit power at the \ac{BS}. 
It is assumed that the {channel between the \ac{BS} and the $k$-th user} is blocked\footnote{{This assumption captures the most challenging and conceptually relevant scenario, where the impact and benefits of RIS-assisted propagation are most clearly observed.
In particular, it highlights coverage extension and blockage mitigation, while reducing the need for network-level optimization measures such as deploying additional base stations. 
As a result, the blocked-user assumption is widely adopted in RIS-related studies, enabling} direct comparison of the proposed technique with \ac{SotA} methods, $e.g.$ \cite{YahyaOJCS2024,ZheyuArxiv2024}.
{Importantly, this assumption is not required by our proposed approach}, since the scattering matrix design depends only on the cascaded BS-{to-}RIS-to-user channels. {Nevertheless, the} overall performance under {this setting} relies on the BS-{to-}RIS channel having a sufficiently high rank, which holds only in rich-scattering, large-antenna setups, motivating an extension of the work under poorly scattered channels with \ac{LoS} links between the \ac{BS} and users.
Such an extension also requires, however, further work on the beamforming algorithm, thus falling beyond the scope of this article, and therefore will be pursued in a follow-up contribution.}, such that a focus is placed on the performance gains due to the \ac{BD-RIS}. 
Thus, the signal $r_k \in \mathbb{C}$ received at {user $k$} can be expressed as
\begin{equation}
\label{eq:rsig}
r_k = \mathbf{h}_k^T\mathbf{\Theta}\mathbf{H}_{\mathrm{TX}}\mathbf{x}+n_k,
\end{equation}
where $n_k \sim \mathcal{CN}(0, N_0)$ denotes \ac{AWGN} with power $N_0$. 

\begin{figure}[H]
\centering
\includegraphics[width=1\linewidth]{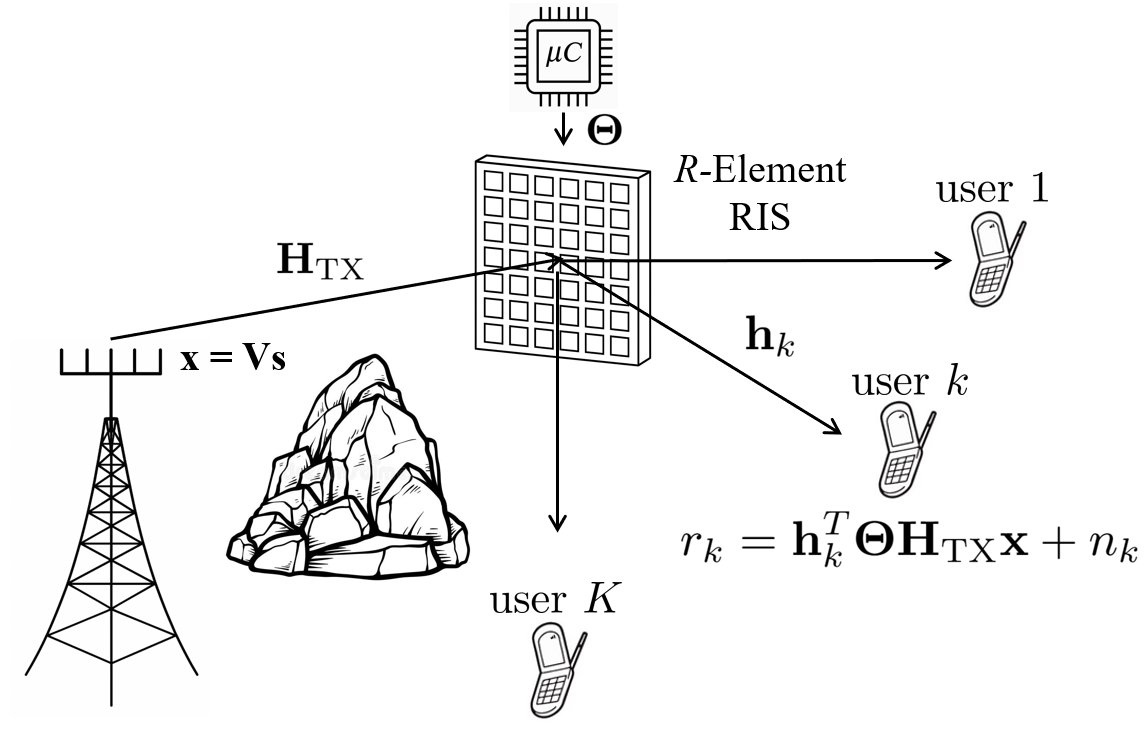}
\caption{Illustration of the system model, where a \ac{BS} with $N$ \ac{TX} antennas serves $K$ single-antenna users through an $R$-Element \ac{RBD-RIS}, without a \ac{LoS} link$^2$ between the BS and the users.}
\label{fig:system model}
\vspace{-2ex}
\end{figure}

The compact vector representation in~\eqref{eq:rsig} is inspired by the formulation used in \cite{YahyaOJCS2024}, and is given as
\begin{subequations}
\begin{gather}
\mathbf{r} = [r_1, r_2, \dots, r_K]^T, \\
\mathbf{n} = [n_1, n_2, \dots, n_K]^T, \\
\mathbf{V} = [\mathbf{v}_1, \mathbf{v}_2, \dots, \mathbf{v}_K], \\
\mathbf{H}_{\mathrm{RX}} = [\mathbf{h}_1, \mathbf{h}_2, \dots, \mathbf{h}_K]^T,  \\
\mathbf{H}_{\mathrm{TX}}= \left[ \mathbf{w}_1, \mathbf{w}_2, \dots, \mathbf{w}_N \right], \\ 
\mathbf{r} = \mathbf{H}_{\mathrm{RX}} \boldsymbol{\Theta} \mathbf{H}_{\mathrm{TX}} \mathbf{x} + \mathbf{n}, \label{eq:rvec}
\end{gather}
\end{subequations}
where $\mathbf{r}\in\mathbb{C}^{K\times 1}$, $\mathbf{H}_{\mathrm{RX}} \in \mathbb{C}^{K \times R}$, $\mathbf{w}_n \in \mathbb{C}^{R\times 1}$ {with} $n\in \{1,2,\dots,N\}$, {and $\mathbf{v}_m \in \mathbb{C}^{N \times 1}$ with $m\in \{1,2,\dots,K\}$}, denote the received signal vector, \ac{BD-RIS}-{to-user $k$} channel matrix, \ac{BS}-{to-}\ac{BD-RIS} channel vector, and beamforming vector, respectively.

\vspace{-3ex}

\subsection{Scattering Matrix Definition}

Depending on the reconfigurable impedance network of the \ac{BD-RIS}, three different architectures are considered, namely:
\begin{enumerate}
\item Single-connected \ac{BD-RIS}: Equivalent to the conventional D-RIS, where each \ac{RE} is independently connected to a single grounded reconfigurable impedance with no interconnections among \acp{RE}.
Consequently, the resulting scattering matrix is inherently diagonal, which naturally satisfies the symmetry constraints, while also being unitary due to the lossless nature{\footnote{{According to microwave network theory, a lossless reconfigurable impedance network must satisfy a unitary scattering matrix condition, thereby ensuring energy conservation \cite{LiCM2024}.
Furthermore, the reciprocity property, which enforces symmetry of the BD-RIS scattering matrix, follows directly from electromagnetic reciprocity and microwave network theory.
In particular, any passive, linear and time-invariant structure exhibits symmetric transmission characteristics, $i.e.$, the coupling between any two ports of the \ac{BD-RIS} is identical in both directions, which is mathematically captured by the symmetric scattering matrix \cite{PozarBook2021}.}}.}
The aforementioned constraints are described as
\vspace{-1ex}
\begin{eqnarray}
&\mathcal{S}_{\text{SC}_1} = \left\{ \mathbf{\Theta} : [\mathbf{\Theta}]_{i,j} = 0,\ \forall i \ne j \right\},&\\
&\mathcal{S}_{\text{SC}_2} = \left\{ \mathbf{\Theta} : \left|[\mathbf{\Theta}]_{i,j}\right| = 1,\ \forall i = j \right\},&
\end{eqnarray}
where $i,j \in \{1, 2, \dots, R\}$.  

\item Fully-connected \ac{BD-RIS}: In this architecture, every \ac{RE} is connected to all others via a reconfigurable impedance network, allowing a wave incident on one element to be reflected by all others.
As a result, the corresponding scattering matrix is full and---similarly to the other architectures---must satisfy both the symmetry and the unitary constraint, which for this case are given as
\vspace{-1ex}
\begin{eqnarray}
&\mathcal{S}_{\text{FC}_1} = 
\left\{ 
  \mathbf{\Theta} : 
  \mathbf{\Theta} = \mathbf{\Theta}^T 
\right\},&\\
&\mathcal{S}_{\text{FC}_2} = 
\left\{ 
\mathbf{\Theta} : 
\mathbf{\Theta} \mathbf{\Theta}^H = \mathbf{I} 
\right\}.&
\end{eqnarray}

\item Group-connected \ac{BD-RIS}: In this intermediate architecture the \acp{RE} are partitioned into $G$ disjoint groups, each consisting of $R_G = \frac{R}{G}$ elements. 
Within each group, all \acp{RE} are interconnected, enabling waves impinging on one element to be reflected by all others in the same group.
Apart from the intra-group connections, no other link exists.
Consequently, the scattering matrix can be represented as a block-diagonal matrix by ordering its elements group-by-group, such that each block corresponds to one group and must independently satisfy the unitary and symmetry constraints. 
Notably, the group-connected architecture generalizes the previous two cases, reducing to the fully-connected case when $G=1$ and to the single-connected case when $G=R$.
Denoting the scattering matrix of groups $g \in \{1, 2, \ldots, G\}$
by $\mathbf{\Theta}_g \in \mathbb{C}^{R_G \times R_G}$, {the overall structure of the group-connected architecture is expressed by
\begin{equation}
  \mathbf{\Theta} = \mathrm{blkdiag}(\mathbf{\Theta}_1, \ldots, \mathbf{\Theta}_G),
\end{equation}
and} the constraints on the group-connected BD-RIS scattering matrix {are given by}
\begin{gather}
\hspace{-5ex} {\mathcal{S}_{\text{GC}_1} \!\!=\!
\left\{ 
  \mathbf{\Theta} : \mathbf{\Theta}_g \!=\! \mathbf{\Theta}_g^T, \forall g
\right\}\!,\!\!\! }
\end{gather}
\begin{equation}
\mathcal{S}_{\text{GC}_2} = 
\left\{ 
  \mathbf{\Theta} : 
  \mathbf{\Theta}_g \mathbf{\Theta}_g^H = \mathbf{I}_{R_G},\; \forall g 
\right\}.
\end{equation}

{Here $\mathbf{\Theta}_g$ is formally defined as
\begin{equation}
\label{eq:group_scattering}
\mathbf{\Theta}_g = \mathbf{\Theta}_{[R_G(g-1)+1\,:\,gR_G,\;\; R_G(g-1)+1\,:\,gR_G]},
\end{equation}
where $g \in \{1,2,\dots,G\}$. 
Figure \ref{fig:bp_bdsris}, seen at the top of the next page, provides an illustrative example demonstrating the impact of beyond-diagonal connectivity on the scattering behavior of the \ac{BD-RIS}.}
\end{enumerate}

\begin{figure*}
\centering
\subfloat[Single-Connected]{
\includegraphics[width=0.32\textwidth]{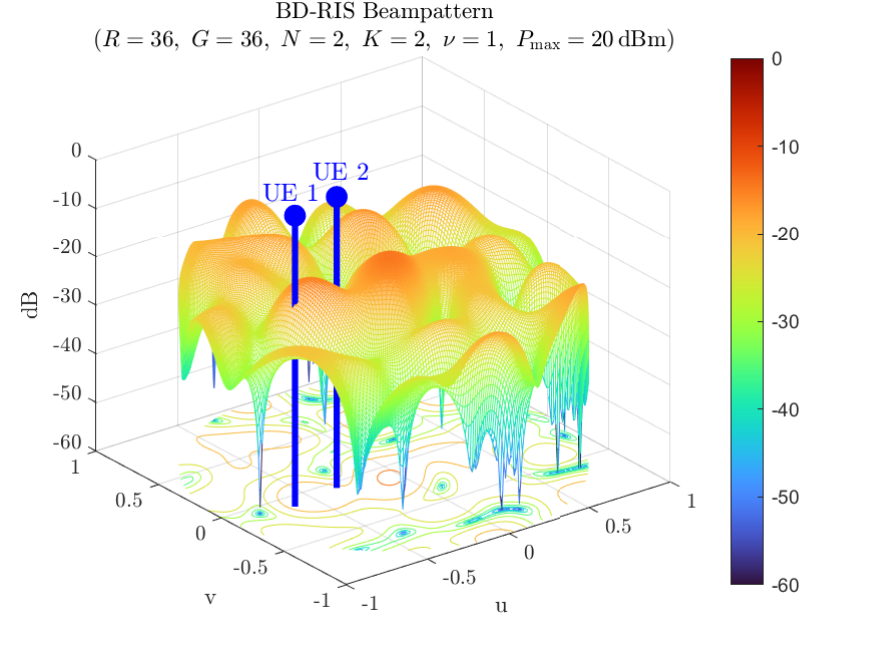}
}
\subfloat[Group-Connected (group size 4)]{
\includegraphics[width=0.32\textwidth]{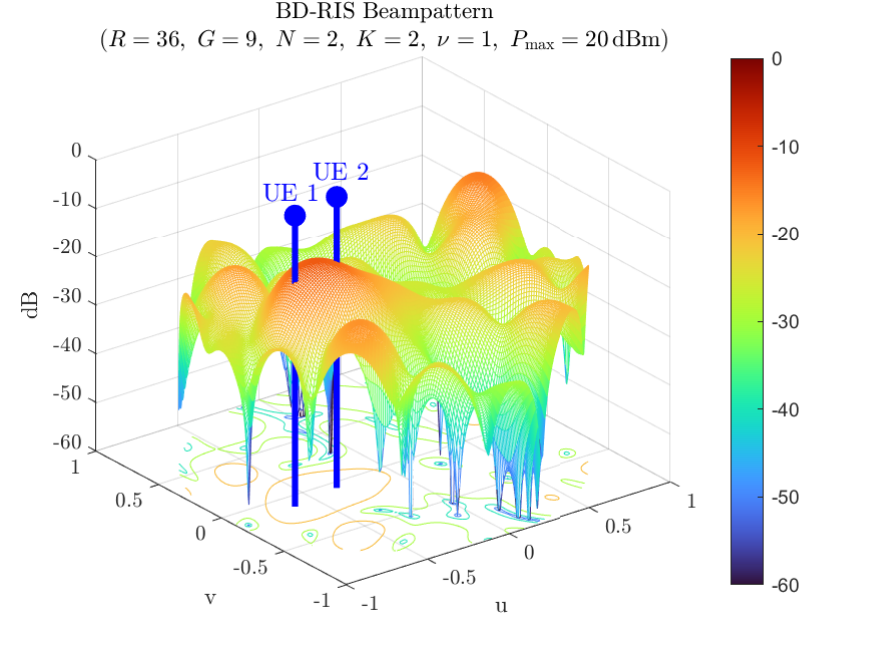}
}
\subfloat[Fully-connected]{
\includegraphics[width=0.32\textwidth]{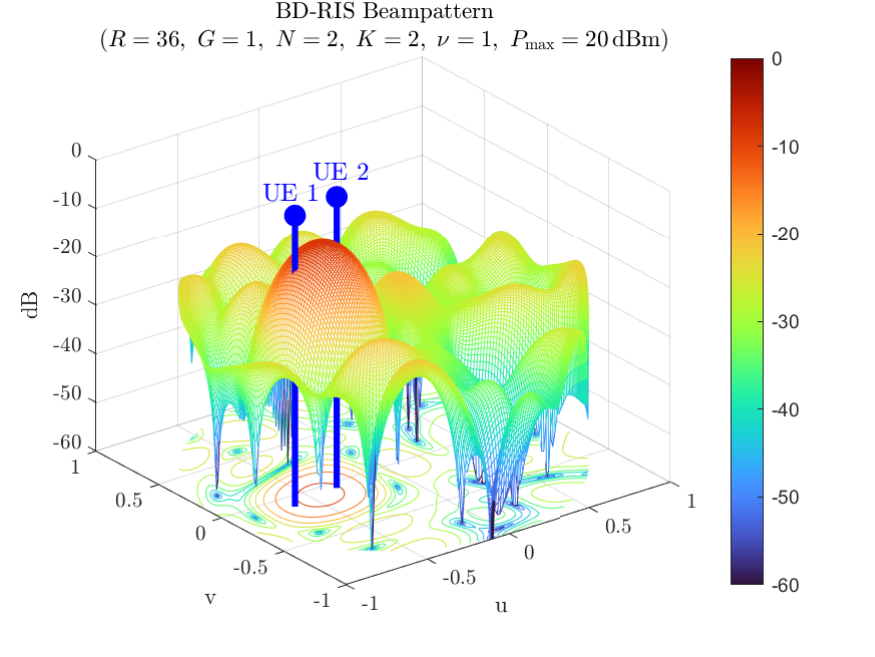}
}\quad\color{white}.
\vspace{-1ex}
\caption{{Beampattern illustrating the responses of a BD-RIS with different levels of connectivity. The plots are shown in terms of the auxiliary quantities $u\triangleq \sin(\phi)\cos(\theta)$ and $v \triangleq \sin(\phi)\sin(\theta)$, where $\phi$ and $\theta$ are the azimuth and elevation angles, swept in the ranges $[-\pi,\pi]$ and $[-\pi/2,\pi/2]$, respectively.
The blue lines correspond to the directions of two users, located at $(\phi_1, \theta_1) = (-40^\circ,10^\circ)$ and $(\phi_2, \theta_2) = (-20^\circ,10^\circ)$.}}
\label{fig:bp_bdsris}
\vspace{-2ex}
\end{figure*}

\vspace{-2ex}
\subsection{Problem Formulation}

Given $\mathbf{\Theta}$ and $\mathbf{V}$, and {assuming perfect \ac{CSI}\footnote{{Although idealized, the assumption of perfect \ac{CSI} is widely adopted in the RIS literature \cite{WuTWC2019,YahyaOJCS2024, ZheyuArxiv2024}, as it enables evaluation of the ``best-case'' performance of the proposed algorithm, while imperfect \ac{CSI} can only degrade performance.
Moreover, in block-static channel conditions where the channel remains constant over the data transmission interval, the channel coefficients can, in principle, be estimated with arbitrarily high accuracy given sufficient training resources. 
In such settings, performance analysis under perfect \ac{CSI} is consistent with the standard block-fading framework, where coding is performed over long coherence intervals.
In contrast, in rapidly varying fading environments, channel estimation errors become unavoidable and must be explicitly incorporated into the system design. 
Extensions to robust formulations under imperfect \ac{CSI} are feasible and have been extensively studied in the literature (see, e.g., \cite{AlmeidaTSP2025
,PapazafeiropoulosOJCS2024}). 
However, such developments fall beyond the scope of the present work.}}}, an achievable rate at {the $k$-th user} is $\log_2(1+\gamma_k)$, where the \ac{SINR}{\footnote{{For the extension to the general \ac{MU}-\ac{MIMO} case, the \ac{SINR} in \eqref{eq:sinr_scalar} would need to be revisited to incorporate receive combining at the user side.}}} is expressed as
\vspace{-1ex}
\begin{equation}
\gamma_k = 
\frac{ \left| \mathbf{h}_k^T \mathbf{\Theta} \mathbf{H}_{\mathrm{TX}} \mathbf{v}_k \right|^2 }
{ \sum\limits_{i \ne k} \left| \mathbf{h}_k^T \mathbf{\Theta} \mathbf{H}_{\mathrm{TX}} \mathbf{v}_i \right|^2 + N_0 }
\ .
\vspace{-1ex}
\label{eq:sinr_scalar}
\end{equation}

The corresponding sum-rate maximization problem is then
\vspace{-0.5ex}
\begin{subequations} \label{prob:P1}
\begin{align}
(\text{P1}): \quad \underset{\mathbf{V},\mathbf{\Theta}}{\mathrm{maximize}} \quad & \sum_k \log_2(1 + \gamma_k) \label{eq:P1_obj} \\
\text{subject to} \quad & \|\mathbf{V}\|_F^2 \leq P_{\mathrm{max}}, \label{eq:P1_power} \\
& \mathbf{\Theta} \in \mathcal{S}_{a_1}, \label{eq:P1_Sa1} \\
& \mathbf{\Theta} \in \mathcal{S}_{a_2}, \label{eq:P1_Sa2}
\end{align}
\end{subequations}
where {the constraint \eqref{eq:P1_power} corresponds to the total transmit power limitation at the \ac{BS}, ensuring that the Frobenius norm of the beamforming matrix $\mathbf{V}$ does not exceed the maximum allowable transmit power $P_{\mathrm{max}}$,} $a \in \{ \text{SC}, \text{FC}, \text{GC} \}$ denotes the connectivity level of the \ac{BD-RIS}, and {$\mathcal{S}_{a_1}$ and $\mathcal{S}_{a_2}$ represent the sets of scattering matrices that satisfy the reciprocal and lossless constraints, respectively.}
\newpage

The optimization problem (P1) is inherently non-convex due to the objective function in \eqref{eq:P1_obj}, and the strong coupling between the \ac{TX} beamforming and the scattering matrix.
Additionally, the constraints in \eqref{eq:P1_power}-\eqref{eq:P1_Sa2} are generally non-convex.
{In particular,} the constraint in \eqref{eq:P1_Sa2} represents a Stiefel manifold condition {\cite{AbsilPSP2008}, since it enforces orthonormality of the columns of the BD-RIS scattering matrix.}

Due to this coupling, a two-stage design is adopted to enhance the sum-rate.
In this article, the first stage of the optimization problem is addressed, namely, the design of the scattering matrix.
{The second stage, concerning the optimization of the transmit beamforming matrices, becomes particularly relevant in more general \ac{MU}-\ac{MIMO} settings where active beamforming plays a larger role in the achievable performance. 
Further discussion on this aspect, as well as on the extension to the general \ac{MU}-\ac{MIMO} case, is provided in the next subsection.}

{
\subsection{A Note on the General MU-MIMO Case}
\label{sec:MUMIMO_Generalization}

\vspace{1ex}

Before we proceed, let us clarify that the method proposed in this article can be directly applied to \ac{MU}-\ac{MIMO} systems with relatively few modifications described below.

Consider a \ac{MU}-\ac{CF-mMIMO} system with $L$ \acp{AP}, each equipped with $N_a$ \ac{TX} antennas, such that the total number of transmit antennas is $N_t = L\times N_a$, and $K$ users equipped with with $M$ \ac{RX} antennas, aided by a \ac{BD-RIS} with $R$ \acp{RE}. 
The corresponding channel linking the $l$-th AP to the $k$-th user can be defined as
\begin{equation}
\mathbf{E}_{l,k} = \mathbf{H}_{\mathrm{RX},k} \mathbf{\Theta} \mathbf{H}_{\mathrm{TX},l} \in \mathbb{C}^{M\times N_a},
\end{equation}
%
where $\mathbf{H}_{\mathrm{RX},k}$ and $\mathbf{H}_{\mathrm{TX},l}$ denote the \ac{BD-RIS}-to-user and \ac{AP}-to-\ac{BD-RIS} channel matrices, respectively.

Defining the equivalent channel to the $k$-th user as the concatenation of the channel matrices from all \acp{AP} to the $k$-th user, we write $\mathbf{E}_k \triangleq [\mathbf{E}_{1,k},\mathbf{E}_{2,k},\dots, \mathbf{E}_{L,k}] \in \mathbb{C}^{M\times N_t}$.
Accordingly, the complex baseband received signal at the $k$-th user, which generalizes \eqref{eq:rsig}, is given by

\quad\\[-3ex]
\begin{equation}
\mathbf{r}_k = \overbrace{\,\mathbf{E}_k \mathbf{V}_k \mathbf{s}_k}^{\text{Intended signal}}\; + \!\!\!\!\!\overbrace{\!\!\!\!\!\sum_{k' \in \mathcal{K}\backslash \{k\}} \!\!\!\!\!\! \mathbf{E}_k \mathbf{V}_{k'} \mathbf{s}_{k'}}^{\text{Downlink inter-user interference}} \!\! +  \mathbf{n}_k \in \mathbb{C}^{M\times 1},
\end{equation}
where $\mathbf{n}_k\sim \mathcal{CN}(0, N_0\mathbf{I})$.

For a given rate $\log_2 \text{det}(\mathbf{I}_M+\mathbf{\Gamma}_k)$, the \ac{SINR} at the $k$-th user, $\mathbf{\Gamma}_k$, generalizes \eqref{eq:sinr_scalar} and is expressed as 
\begin{subequations}
\begin{equation}
\label{eq:SINR_MIMO}
\mathbf{\Gamma}_k = \mathbf{V}_k^\mathrm{H} \mathbf{E}_k^\mathrm{H} \mathbf{\Psi}_k^{-1} \mathbf{E}_k \mathbf{V}_k \in \mathbb{C}^{M \times M},
\end{equation}
where
\begin{equation}
\label{eq:interf}
\mathbf{\Psi}_k \triangleq  \sum\limits_{k'\in \mathcal{K}\backslash \{k\}} \mathbf{E}_{k} \mathbf{V}_{k'} \mathbf{V}_{k'}^\mathrm{H} \mathbf{E}^\mathrm{H}_{k} + N_0  \mathbf{I}_M \mathbb{C}^{M \times M}
\end{equation}
\end{subequations}
is the interference-plus-noise covariance matrix, and $\mathbf{V}_k$ denotes the beamforming matrix for the $k$-th user.

From the above, it is evident that the proposed framework applies directly to general \ac{MIMO} scenarios, as previously mentioned.
We emphasize, however, such a setting would obscure the isolated impact of the \ac{BD-RIS} scattering matrix design.
In addition, a comprehensive treatment of the \ac{MIMO} case would not be complete without considering also the design of corresponding beamformers.
For these reasons, we leave the joint design of \ac{TX} beamforming and \ac{BD-RIS} scattering matrix for future work, and continue hereafter addressing the \ac{MU}-\ac{MISO} case, for simplicity, but without loss of generality.
}

\section{Stage 1: Scattering Matrix Design}
\label{sec:smdesign}

We will solve (P1) in an iterative manner.
To first design the scattering matrix, the impact of the \ac{BS} beamforming is fixed by initializing it with a \ac{MMSE}-based beamformer{\footnote{{Perfect CSI is assumed as discussed previously. 
However, MMSE precoding has been shown to be robust under imperfect CSI \cite{SavauxTVT2025}, so this assumption does not represent a fundamental weakness of the considered framework.}}}. 
Accordingly, \eqref{eq:rvec} is rewritten as
\begin{equation}
\mathbf{r}=\mathbf{H}_{\mathrm{RX}}\mathbf{\Theta}\mathbf{H}_{\mathrm{TX}}\mathbf{V}\mathbf{s}+\mathbf{n}
\triangleq \mathbf{H}_{\mathrm{RX}}\mathbf{\Omega}\mathbf{V}\mathbf{s}+\mathbf{n}
\triangleq \mathbf{E}\mathbf{V}\mathbf{s}+\mathbf{n},
\end{equation}
where $\mathbf{\Omega} = \mathbf{\Theta}\mathbf{H}_{\mathrm{TX}} = [\boldsymbol{\omega}_1,\dots, \boldsymbol{\omega}_N] \in \mathbb{C}^{R\times N}$ and $\mathbf{E} = \mathbf{H}_{\mathrm{RX}}\mathbf{\Omega}= [\mathbf{e}_1,\dots,\mathbf{e}_K]^T \in \mathbb{C}^{K\times N}$ denotes the equivalent channel.
\newpage

As a consequence, the \ac{SINR} at {the $k$-th user} during the design of the scattering matrix is given as
\begin{equation}
\label{eq:SINR}
\hspace{-1ex}
\gamma_k = 
\frac{ \left| \mathbf{h}_k^T \boldsymbol{\Omega} \mathbf{v}_k\right|^2}
{\sum\limits_{i \ne k} \left| \mathbf{h}_k^T \boldsymbol{\Omega} \mathbf{v}_i \right|^2 + N_0}
=
\frac{\left| \mathbf{e}_{k} \mathbf{v}_k \right|^2}
{\sum\limits_{i \ne k} \left| \mathbf{e}_{k} \mathbf{v}_i \right|^2 + N_0}.
\end{equation}

Making use of \eqref{eq:SINR}, the optimization problem {of concern can be} formulated as
\begin{subequations} \label{prob:P2}
\begin{align}
  (\text{P2}): \quad \underset{\mathbf{\Theta}}{\mathrm{maximize}} \quad & \sum_k \log_2(1 + \gamma_k)  \\
  \text{subject to} \quad & \mathbf{\Theta} \in \mathcal{S}_{a_1},  \\
  & \mathbf{\Theta} \in \mathcal{S}_{a_2}.
\end{align}
\end{subequations}

\subsection{Solution for the Group-Connected Architecture}

As the group-connected \ac{BD-RIS} architecture is the most {general version} considered, it can provide a generic solution to (P2), which in turn can be rewritten as
\begin{subequations}  
\begin{align}
 (\text{P2a}): \quad \underset{\mathbf{\Theta}}{\mathrm{maximize}} \quad & \sum_k \log_2(1 + \gamma_k)  \\
\text{subject to} \quad & \mathbf{\Theta}_g = \mathbf{\Theta}_g^T,  \\
 & \mathbf{\Theta}_g\mathbf{\Theta}_g^H = \mathbf{I}_{R_G},
\end{align}
\end{subequations} 
where $g=\{1,2,\dots,G\}$.

For simplicity, the sum-rate over all users is denoted as
\begin{equation}
\eta = \sum_k \eta_k, \quad \text{with} \quad \eta_k = \log_2(1 + \gamma_k).
\end{equation}

As previously noted, the optimization problem (P2a) is inherently non-convex due to the \ac{SINR} being a ratio of quadratic terms of $\mathbf{\Theta}$.
The unitary equality constraint further complicates the problem due to its non-convex nature.

To address this matter, manifold optimization is employed by treating $\mathbf{\Theta}_g$ as a point on the Stiefel manifold.
Meanwhile, the symmetry constraint restricts $\mathbf{\Theta}_g$ to the manifold of symmetric complex matrices, which forms a linear subspace and thus also a smooth manifold.
To simultaneously enforce both the unitary and symmetry constraints, an intersection of the two sets is needed.
However, as this intersection is challenging to handle directly, a workaround approach is considered by modifying the objective function to add a penalty that promotes symmetry.

Accordingly, the optimization problem then becomes
\begin{subequations}  
\begin{align}
\hspace{-1.35ex}
 (\text{P3}): \quad \underset{\mathbf{\Theta}}{\mathrm{maximize}} \quad & \!\!\sum_k \log_2(1 + \gamma_k)
  - \nu \bigl\| \boldsymbol{\Theta} - \boldsymbol{\Theta}^T \bigr\|_F^2
\label{eq:P3_obj}\\
\text{subject to} \quad & \!\!\mathbf{\Theta}_g\mathbf{\Theta}_g^H = \mathbf{I}_{R_G},\label{eq:P3_const}
\end{align}
\end{subequations}
where $\nu \in \mathbb{R}$ denotes a nonnegative weight{\footnote{{The penalty parameter $\nu$ is selected empirically based on preliminary simulations. 
While systematic hyperparameter optimization techniques could be employed (see, e.g., \cite{Boukari1995,Pedregosa2016,Bischl2023
}), the focus of this work is on the scattering matrix design framework itself rather than on optimal regularization tuning. 
Moreover, in the following sections, the proposed method is shown to outperform existing \ac{SotA} approaches even with empirically chosen $\nu$, and further performance improvements may be achieved through adaptive or data-driven tuning strategies.}}} and $\bigl\| \boldsymbol{\Theta} - \boldsymbol{\Theta}^T \bigr\|_F^2
$ is used to obtain a quantitative measure of how far the matrix is from being symmetric.

Having defined the optimization problem (P3), its solution can be obtained numerically using the Manopt toolbox \cite{BoumalJMLR2014}.
However, {measures can be taken to reduce the complexity of the optimization problem.
In particular, to address the challenges arising in part from the non-convexity of the sum-rate function, the objective can be transformed into a convex form using \ac{FP} techniques \cite{
ShenTSP2018_1, ShenTSP2018_2}.


{To this extent, we apply the \ac{LDT} to the sum-rate objective in (P3), which leads to the equivalent formulation
\vspace{-1ex}
\begin{equation}
  \label{eq:ldt_sr}
  \bar{\eta}_k = \log_2(1+\tau_k) - \frac{\tau_k}{\ln(2)} + \frac{1+\tau_k}{\ln(2)}\cdot\frac{\left| \mathbf{e}_{k} \mathbf{v}_k \right|^2}
{\sum\limits_{i} \left| \mathbf{e}_{k} \mathbf{v}_i \right|^2 + N_0},
\vspace{-1ex}
\end{equation}
where the auxiliary variable $\tau_k \in \mathbb{C}$ denotes the Lagrange multiplier, s.t. $\tau_k = \gamma_k$. 

The fractional term in the majorized rate expression $\bar{\eta}_k$ given in \eqref{eq:ldt_sr} can be further convexified by applying the \ac{QT}, resulting in the following convex reformulation of the sum-rate objective:
\begin{align}
  \label{eq:eta_hat}
  \hat{\eta}_k = &  \log_2(1+\tau_k) - \frac{\tau_k}{\ln(2)} \\
  & + \frac{1+\tau_k}{\ln(2)} \Big[ 2\Re\{y_k^{\star}\mathbf{e}_k\mathbf{v}_k\}-\left|y_k\right|^2\big(\sum_{i}\left|\mathbf{e}_k\mathbf{v}_i\right|^2+N_0\big)\Big], \nonumber
\end{align}
where $y_k \in \mathbb{C}$ is an auxiliary variable introduced by the \ac{QT}, given as
\begin{equation}
  y_k = \frac{\mathbf{e}_k\mathbf{v}_k}{\sum\limits_{i}\left|\mathbf{e}_k\mathbf{v}_i\right|^2+N_0}.
\end{equation}

In this manner, the optimization problem (P3) can be reformulated as
\begin{subequations}  
\begin{align}
\hspace{-1.35ex}
 (\text{P3a}): \quad \underset{\mathbf{\Theta}}{\mathrm{maximize}} \quad & \!\!\sum_k \hat{\eta}_k
  - \nu \bigl\| \boldsymbol{\Theta} - \boldsymbol{\Theta}^T \bigr\|_F^2
\label{eq:P3a_obj}\\
\text{subject to} \quad & \!\!\mathbf{\Theta}_g\mathbf{\Theta}_g^H = \mathbf{I}_{R_G},\label{eq:P3a_const}
\end{align}
\end{subequations}
where, for convenience, the modified sum-rate objective function in \eqref{eq:P3a_obj}, including the penalty term, is denoted as $\breve{\eta}_k$.

}


{In addition, to alleviate the high computational burden associated with Manopt, a tailored} \ac{CGA} algorithm \cite{ShewchukCMU1994} is derived in the sequel.
This algorithm employs a closed-form expression for the gradient of the objective function, as given in \eqref{eq:P3a_obj}, thereby enabling a more efficient implementation of the scattering matrix design.

To derive a solution for the general group-connected case, we define} the vector $\mathbf{h}_k^{(g)T} = \mathbf{H}_{RX\left[k, R_G(g-1)+1:gR_G\right]}$ and the matrix ${\mathbf{H}^{(g)}_{\mathrm{TX}}} = \mathbf{H}_{\mathrm{TX}\left[R_G(g-1)+1:gR_G, 1:N\right]}$, {which} represent the channel components associated with the $g$-th group.
Specifically, $\mathbf{h}_k^{(g)T} \in \mathbb{C}^{1 \times R_G}$ denotes the \ac{BD-RIS}-{to-user $k$} channel vector, and ${\mathbf{H}^{(g)}_{\mathrm{TX}}} \in \mathbb{C}^{R_G \times N}$ denotes the \ac{BS}-{to-}\ac{BD-RIS} channel matrix corresponding to the group size $R_G$, respectively.
Within the $g$-th group dimension, the scattering matrix $\mathbf{\Theta}_g \in \mathbb{C}^{R_G \times R_G}$ models the group-specific scattering behavior.
Accordingly, the group-wise equivalent channel is defined as $\mathbf{e}_k^{(g)} = \mathbf{h}_k^{(g)T} \mathbf{\Theta}_g {\mathbf{H}^{(g)}_{\mathrm{TX}}}$, where $\mathbf{e}_k^{(g)} \in \mathbb{C}^{1\times N}$. 



{As such, the gradient of the objective function \eqref{eq:P3a_obj} with respect to each group $\mathbf{\Theta}_g$ is computed as
\begin{equation}
\label{eq:grad}
\nabla_{\mathbf{\Theta}_g}\!\breve{\eta}  =  \nabla_{\mathbf{\Theta}_g} \! \Big( \!  \sum_k \hat{\eta}_k - \nu \left\| \mathbf{\Theta}_g - \mathbf{\Theta}_g^T \right\|_F^2 \! \Big),
\vspace{-1ex}
\end{equation}}
where the symmetry penalty term can be rewritten as
\begin{align}
\left\| \mathbf{\Theta}_g - \mathbf{\Theta}_g^T \right\|_F^2 & = \mathrm{Tr}\{(\mathbf{\Theta}_g-\mathbf{\Theta}_g^T)^H(\mathbf{\Theta}_g-\mathbf{\Theta}_g^T)\} \label{eq:trace_term} \\
& =\mathrm{Tr}\{(\mathbf{\Theta}_g^H-\mathbf{\Theta}_g^{\star})(\mathbf{\Theta}_g-\mathbf{\Theta}_g^T)\} \nonumber  \\ 
& =\mathrm{Tr}\{\mathbf{\Theta}_g^H\mathbf{\Theta}_g-\mathbf{\Theta}_g^H\mathbf{\Theta}_g^T-\mathbf{\Theta}_g^{\star}\mathbf{\Theta}_g+\mathbf{\Theta}_g^{\star}\mathbf{\Theta}_g^T\}.
\nonumber
\end{align}

The gradient for half of the terms in \eqref{eq:trace_term} can be effectively found in {\cite[eq.~(248)]{Petersen2008}}, namely
\begin{equation}
\nabla_{\mathbf{\Theta}_g}\mathrm{Tr}\{\mathbf{\Theta}_g^H\mathbf{\Theta}_g\} = \nabla_{\mathbf{\Theta}_g}\mathrm{Tr}\{\mathbf{\Theta}_g^{\star}\mathbf{\Theta}_g^T\}^T = 2\mathbf{\Theta}_g.
\end{equation}

Conversely,
\begin{equation}
 \nabla_{\mathbf{\Theta}_g}\mathrm{Tr}\{\mathbf{\Theta}_g^H\mathbf{\Theta}_g^T\}=\nabla_{\mathbf{\Theta}_g}\mathrm{Tr}\{\mathbf{\Theta}_g^{\star}\mathbf{\Theta}_g\}^T,
\end{equation}
where the gradient is obtained as
\begin{equation}
\nabla_{\mathbf{\Theta}_g}\mathrm{Tr}\{\mathbf{\Theta}_g^{\star}\mathbf{\Theta}_g\} = 2\frac{\partial\mathrm{Tr}\{\mathbf{\Theta}_g^{\star}\mathbf{\Theta}_g\}}{\partial\mathbf{\Theta}_g^{\star}} = 2\mathbf{\Theta}_g^T,
\label{eq:Transpose}
\end{equation}
using the result from Lemma~\ref{lem:gradAconjA}.

The reader is referred to Appendix \ref{app:lemma1} for further details.

Based on the above, the gradient of $\left\| \mathbf{\Theta}_g - \mathbf{\Theta}_g^T \right\|_F^2$ with respect to $\mathbf{\Theta}_g$ is derived as
\begin{equation}
\nabla_{\mathbf{\Theta}_g}\left\| \mathbf{\Theta}_g - \mathbf{\Theta}_g^T \right\|_F^2 = 4\left(\mathbf{\Theta}_g-\mathbf{\Theta}_g^T\right).
\end{equation}

{Omitting the constant terms of the equivalent sum-rate objective after the \ac{FP} transformations allows for the gradient to be expressed as
\vspace{-1ex}
\begin{equation}
\label{eq:gradD}
\nabla_{\mathbf{\Theta}_g} \hat{\eta}_k \!=\!  \nabla_{\mathbf{\Theta}_g} \!\Big( 2\Re\{y_k^{\star}\mathbf{e}_k\mathbf{v}_k\}\!-\!\left|y_k\right|^2\!\big(\sum_{i}\left|\mathbf{e}_k\mathbf{v}_i\right|^2\!+\!N_0\big) \Big),
\vspace{-1ex}
\end{equation}
where the gradient with respect to both terms is given as 
\begin{equation}
  \label{eq:term1_grad}
  \nabla_{\mathbf{\Theta}_g} 2\Re\{y_k^{\star}\mathbf{e}_k\mathbf{v}_k\} = 2\left(y_k^{\star}\frac{\partial \mathbf{e}_{k}\mathbf{v}_k}{\partial\mathbf{\Theta}_g}\right)^{\star},
\end{equation}
and
\begin{equation}
  \label{eq:term2_grad}
  \nabla_{\mathbf{\Theta}_g} \left|y_k\right|^2\big(\sum_{i}\left|\mathbf{e}_k\mathbf{v}_i\right|^2+N_0\big) = 2 \left|y_k\right|^2\sum_{i}\left(\big(\mathbf{e}_k\mathbf{v}_i\big)^{\star}\frac{\partial \mathbf{e}_{k}\mathbf{v}_i}{\partial\mathbf{\Theta}_g}\right)^{\star},
\end{equation}
\color{black}
using the result obtained in Lemma~\ref{lem:lem2}. 
The reader is referred to Appendix \ref{app:lemma2} for further details.}

Finally, the partial derivatives in \eqref{eq:term1_grad} and \eqref{eq:term2_grad} are {given by}
\begin{align}
\hspace{-1.25ex}
\frac{\partial \mathbf{e}_{k}\mathbf{v}_l}{\partial\mathbf{\Theta}_g} &= \frac{\partial}{\partial\mathbf{\Theta}_g}\! \left(\mathbf{h}_k^{(g)T} \mathbf{\Theta}_g \mathbf{W}^{(g)}\mathbf{v}_l\right)\! = \mathbf{h}_k^{(g)}(\mathbf{W}^{(g)}\mathbf{v}_{l})^T\!,
\end{align}
{where $l \in\{i,k\}$.}

Accordingly, the gradient of the objective function in \eqref{eq:grad} is reformulated as shown in \eqref{eq:fullgradbf}.
For convenience, \eqref{eq:fullgrad} presents the gradient of the objective function assuming a simple power allocation matrix used for beamforming\footnote{This assumption can only be made for fully-loaded systems ($K$ = $N$).}, $i.e.$, $\mathbf{V} = \mathrm{diag}(\sqrt{v}_1, \dots, \sqrt{v}_K)$, where $v_k \in \mathbb{R}, \forall k$.
{For convenience,} both gradients are provided at the top of {the following} page.

\begin{figure*}[t]
\begin{align}
  {
\nabla_{\mathbf{\Theta}_g}\breve{\eta} =  \sum_k \frac{1+\tau_k}{\ln(2)}\Big[2\left(y_k^{\star}\mathbf{h}_k^{(g)}\big(\mathbf{W}^{(g)}\mathbf{v}_k \big)^T\right)^{\star} - 2 \left|y_k\right|^2\sum_{i}\left(\big(\mathbf{e}_k\mathbf{v}_i\big)^{\star}\mathbf{h}_k^{(g)}\big(\mathbf{W}^{(g)}\mathbf{v}_i\big)^{T}\right)^{\star}\Big] \!- \!
    4\left(\mathbf{\Theta}_g-\mathbf{\Theta}_g^T\right). }
\label{eq:fullgradbf}
\end{align}
\begin{center}
\rule{0.5\textwidth}{0.1pt}
\begin{align}
  {
\nabla_{\mathbf{\Theta}_g}\breve{\eta} = \sum_k \frac{1+\tau_k}{\ln(2)}\Big[2v_k\big(y_k^{\star}\mathbf{h}_k^{(g)}\mathbf{w}_k^{(g)T}\big)^{\star} -2\left|y_k\right|^2\sum_{i}v_i^2\big(e_{k,i}^{\star}\mathbf{h}_k^{(g)}\mathbf{w}_i^{(g)T}\big)^{\star}\Big] - 4\nu(\mathbf{\Theta}_g \!- \mathbf{\Theta}_g^T). }
\label{eq:fullgrad}
\end{align}
\rule{1\textwidth}{0.1pt}
\end{center}
\vspace{-2ex}
\end{figure*}


The constraint of the optimization problem, given in \eqref{eq:P3a_const}, defines a complex Stiefel manifold of dimension $R_g^2$.
For clarity, a manifold is a topological space that locally resembles the Euclidean space, but may have a globally curved or more complex structure. 
Such manifolds naturally enforce constraints and symmetries, including orthogonality or unitary conditions \cite{JuniorJCS2024}.

The complex Stiefel manifold {corresponding to square blocks (\( p = n = R_G \)) is given by}
\vspace{-0.5ex}
\begin{equation}
\label{eq:stiefel}
\mathrm{St}(R_G,R_G;\mathbb{C})=\{\mathbf{\Theta}_g\in\mathbb{C}^{R_G\times R_G}:
\mathbf{\Theta}_g\mathbf{\Theta}_g^H=\mathbf{I}_{R_G}\}.
\vspace{-0.5ex}
\end{equation}

Since the general group-connected \ac{BD-RIS} scattering matrix is block-diagonal, it can be expressed as
\vspace{-0.5ex}
\begin{equation}
\hspace{-0.44ex} \boldsymbol{\Theta} = \mathrm{blkdiag}(\boldsymbol{\Theta}_1, \boldsymbol{\Theta}_2, \dots, \boldsymbol{\Theta}_G), \: \boldsymbol{\Theta}_g \in \mathrm{St}(R_G, R_G; \mathbb{C}),
\vspace{-0.5ex}
\end{equation}
where $g \in \{1,2,\dots,G\}$.

The usage of the \ac{CGA} algorithm is motivated by the rich geometric structure of Riemannian manifolds, which allows well-defined gradients of cost functions.
Additionally, optimization over a manifold is locally similar to that in the Euclidean space.
As a result, optimization techniques developed for the Euclidean spaces, such as the gradient descent and trust-region methods, have corresponding formulations on manifolds \cite{YuICCC2019}.

The algorithm proceeds through a sequence of retraction operations $\mathrm{R}(\cdot,\cdot)$, followed by a projection $\mathrm{T}(\cdot,\cdot)$ onto the tangent space.
The retraction operation is expressed as
\vspace{-0.5ex}
\begin{equation}
\mathrm{R}_\mathbf{\Theta}(\mathbf{\Theta}, \boldsymbol{\Xi}) = \mathcal{Q}(\mathbf{\Theta} + \alpha \boldsymbol{\Xi}),
\vspace{-0.5ex}
\end{equation}
where $\alpha$ denotes the step size, $\mathbf{\Theta}$ the current point on the manifold, and $\boldsymbol{\Xi}$ the ascent direction in the tangent space.

In the above, the operator $\mathcal{Q}(\cdot)$ returns the Q-factor from the QR decomposition of  $\mathbf{\Theta} + \alpha \boldsymbol{\Xi}$, ensuring that $\mathrm{R}_\mathbf{\Theta}(\mathbf{\Theta} , \boldsymbol{\Xi}) \in \mathrm{St}(R, R; \mathbb{C})$ \cite{HiroyukiCDC2014}.
Thus, the retraction function maps a point in the tangent space back onto the manifold.

The retraction function for each block is defined as
\vspace{-0.5ex}
\begin{equation}
\mathbf{Y}_g = \mathbf{\Theta}_g + \alpha \boldsymbol{\Xi}_g, \: [\mathcal{Q}_g, \mathcal{R}_g] = \mathcal{QR}(\mathbf{Y}_g), \: \mathbf{\Theta}_g^{\text{new}} = \mathcal{Q}_g.
\vspace{-0.5ex}
\end{equation}

Complementary to the retraction operation, the tangent projection function describes the feasible directions of movement at a point on a manifold, represented by tangent vectors.
All tangent vectors at a point form the tangent space, which includes every possible direction the point can move to.

Furthermore, the tangent space can be regarded as a Euclidean space, where a particular tangent vector, denoted as the Riemannian gradient, indicates the direction of steepest descent of the objection function \cite{AbsilPSP2008}.

In light of the above, it is initially required for the Euclidean gradient of the objective function to be computed.
The Riemannian gradient is then obtained by projecting this Euclidean gradient, as given in \eqref{eq:fullgradbf} and \eqref{eq:fullgrad}, onto the tangent space. 
This projection is expressed as
\vspace{-0.5ex}
\begin{equation}
\label{eq:tangentproj_full}
\mathrm{T}_{\mathbf{\Theta}}( \nabla_{\mathbf{\Theta}} \eta , \mathbf{\Theta}) = \nabla_{\mathbf{\Theta}} \eta - \mathbf{\Theta} \cdot \frac{\mathbf{\Theta}^H\nabla_{\mathbf{\Theta}} \eta + (\nabla_{\mathbf{\Theta}} \eta )^H \mathbf{\Theta}}{2}.
\vspace{-0.5ex}
\end{equation}

Straightforwardly, for each group $g$, the projection onto the tangent space is given as
\vspace{-0.5ex}
\begin{equation}
\label{eq:tangentproj}
\mathrm{T}_{\mathbf{\Theta}_g}( \nabla_{\!\mathbf{\Theta}_g} \eta , \mathbf{\Theta}_g) \!=\! \nabla_{\!\mathbf{\Theta}_g} \eta  - \mathbf{\Theta}_g \cdot \frac{\mathbf{\Theta}_g^H\nabla_{\mathbf{\Theta}_g} \eta + (\nabla_{\mathbf{\Theta}_g} \eta )^H \mathbf{\Theta}_g}{2}. 
\vspace{-0.5ex}
\end{equation}

After the optimization, $i.e.$, following the sequence of retraction and tangent projection operations, the method from \cite{FangCL2024} is applied to enforce the symmetry and unitary constraints, since adding a penalty serves only as a workaround approach.
The full procedure implemented in this article is outlined in Algorithm~\ref{alg:cga}.

\begin{algorithm}[H]
\caption{Proposed CGA for {R}BD-RIS Optimization}
\label{alg:cga}
\begin{algorithmic}[1]
\REQUIRE $\mathbf{H}_{\mathrm{TX}}, \mathbf{H}_{\mathrm{RX}}, P_{\mathrm{max}}, N, N_0, G, \nu, I, \epsilon$
\ENSURE Optimized RIS matrix $\mathbf{\Theta}^{\mathrm{opt}}$
\INITIALIZE $\mathbf{\Theta}^{(0)}$ as a random block-diagonal unitary symmetric matrix
\STATE {Compute initial FP auxiliary variables $\tau_k$ and $y_k$, $\forall k$}
\STATE Compute initial objective ${\breve{\eta}}(\mathbf{\Theta}^{(0)})$ and Riemannian gradient $\mathbf{r}^{(0)} = \mathrm{T}_{\mathbf{\Theta}}(\nabla_{\mathbf{\Theta}}{\breve{\eta}}(\mathbf{\Theta}^{(0)}), \mathbf{\Theta}^{(0)})$ {using \eqref{eq:tangentproj_full}}
\vspace{0.125em}
\STATE Set initial search direction $\boldsymbol{\Xi}^{(0)} = -\mathbf{r}^{(0)}$
\FOR{$i = 0$ to $I$}
\IF{$\langle \mathbf{r}^{(i)}, \boldsymbol{\Xi}^{(i)} \rangle \leq 0$}
\vspace{0.125em}
    \STATE Set $\boldsymbol{\Xi}^{(i)} = \mathbf{r}^{(i)}$
\ENDIF
\STATE Compute $\alpha^{(i)}$ via Armijo line search and update $\mathbf{\Theta}^{(i+1)}$
\vspace{-2.5ex}
\STATE {Compute new FP auxiliary variables $\tau_k$ and $y_k$, $\forall k$}
\STATE Compute new objective ${\breve{\eta}}(\mathbf{\Theta}^{(i+1)})$, {using \eqref{eq:P3a_obj}}
\STATE Compute new Riemannian gradient $\mathbf{r}^{(i+1)}$, using \eqref{eq:tangentproj}
\vspace{0.125em}
\STATE Compute $\beta^{(i)} = \max\left(0, \frac{\langle \mathbf{r}^{(i+1)}, \mathbf{r}^{(i+1)} - \mathbf{r}^{(i)} \rangle}{\langle \mathbf{r}^{(i)}, \boldsymbol{\Xi}^{(i)} \rangle}\right)$
\vspace{0.125em}
\STATE Update direction: $\boldsymbol{\Xi}^{(i+1)} = -\mathbf{r}^{(i+1)} + \beta^{(i)} \boldsymbol{\Xi}^{(i)}$
\IF{$|\eta(\mathbf{\Theta}^{(i+1)}) - \eta(\mathbf{\Theta}^{(i)})| < \epsilon$}
    \STATE \textbf{break}
\ENDIF
\ENDFOR
\FOR{$g = 1$ to $G$}
\STATE Extract block $\mathbf{\Theta}_g$ from $\mathbf{\Theta}^{(i+1)}$
\STATE Symmetrize: $\mathbf{\Theta}_{\mathrm{sym}} = \frac{1}{2}(\mathbf{\Theta}_g + \mathbf{\Theta}_g^{T})$
\STATE Perform SVD: $\mathbf{\Theta}_{\mathrm{sym}} = \mathbf{U}\mathbf{\Sigma}\mathbf{V}^H$
\STATE Set $\mathbf{\Theta}^{\mathrm{opt}}_{g} = \mathbf{U}\mathbf{V}^H$
\ENDFOR
\STATE Assemble full $\mathbf{\Theta}^{\mathrm{opt}}$ as block diagonal of all $\mathbf{\Theta}_{g}^{\mathrm{opt}}$
\RETURN $\mathbf{\Theta}_{\mathrm{opt}}$
\end{algorithmic}
\end{algorithm}
\vspace{-2ex}

%
%
%
%

%
The algorithm begins by initializing a random unitary symmetric block-diagonal matrix $\mathbf{\Theta}^{(0)}$ and computing the initial objective function along with its gradient. 
In each iteration, these quantities are updated according to the current values of the scattering matrix until convergence is achieved.
The step size $\alpha$ is determined using an Armijo line search \cite{AbsilPSP2008, Nocedal2009}, which ensures that there is a sufficient ascent.
Furthermore, the parameter $\beta$ denotes the \ac{CG} momentum term, used to enhance the convergence speed of the algorithm. 

Among various options, including the Fletcher-Reeves, Polak Ribi\`ere, and Hestenes-Stiefel formula, this algorithm incorporates the Polak-Ribi\`ere formula as described in \cite{ShewchukCMU1994, Nocedal2009}.
Finally, the scattering matrix resulting from the series of retraction and tangential projection operations is further processed using the method {described in \cite{FangCL2024}}, ensuring that $\mathbf{\Theta}_g^{\mathrm{opt}}$ satisfies both the symmetric and unitary constraints.

%
%

\textit{Computational Complexity:} The computational complexity of Algorithm~\ref{alg:cga} is primarily influenced by the Armijo line search. 
Within the line search, two operations dominate the complexity, namely a triple matrix multiplication and a QR decomposition-based retraction.
Considering the general case of the group-connected architecture, the QR decomposition for each group, of size $R/G$, has a complexity of $\mathcal{O}\big((R/G)^3\big)$, while the triple matrix multiplication has a complexity of $\mathcal{O}\big(KR(R+N)\big)$.
Since the QR decomposition is applied independently to each group, the total complexity of the retraction operation across all groups and users is given as $\mathcal{O}\big(KG(R^3/G^3)\big)$.

\begin{table*}
\centering
\small
\renewcommand{\arraystretch}{1.2}
\caption{Computational Complexity of BD-RIS Scattering Matrix Design}
\begin{tabular}{|>{\centering\arraybackslash}p{0.15\textwidth}|>{\centering\arraybackslash}p{0.375\textwidth}|>{\centering\arraybackslash}p{0.375\textwidth}|}
\hline
\multirow{2}{*}{\centering BD-RIS Architecture} & \multicolumn{2}{c|}{Computational Complexity} \\ \cline{2-3}
& Interference Nulling SotA Alg. \cite{YahyaOJCS2024} & Proposed Alg. \ref{alg:cga} \\ \hline
Single-connected & $\mathcal{O}(tR + K^6 + K^4R + K^2R^2 + K^2R)$& $\mathcal{O}(IL(KR + K^2R^2 + K^2RN))$ \\ \hline
Group-connected & $\mathcal{O}(tG\!R_G^3 \!+ \!tG\!R_G^2\! +\! K^6\!\! +\! K^4G\!R_G^2\! +\! K^2\!G^2\!R_G^4\! +\! K^2\!G\!R_G^2)$ & $\mathcal{O}(IL(KGR_G^3 + K^2R^2 + K^2RN))$ \\ \hline
Fully-connected & $\mathcal{O}(tR^3+tR^2+K^6 +K^4R^2 + K^2R^4 +K^2R^2)$& $\mathcal{O}(IL(KR^3 + K^2R^2 + K^2RN))$ \\ \hline
\end{tabular}
\label{tab:complexity}
\vspace{-3ex}
\end{table*}

In contrast, the triple matrix multiplication is performed on the full system, $i.e.$, not on a per-group basis, and yields a total complexity for all users of $\mathcal{O}\big(K(KR^2+KRN)\big)$.
With $L$ Armijo iterations, these contributions accumulate, giving an overall line search complexity of 
\vspace{-0.5ex}
\begin{equation}
\mathcal{O}\big(LKG(R^3/G^3) + LK(KR^2+KRN)\big).
\vspace{-0.5ex}
\end{equation}

The gradient computation consists of several terms, resulting in a total complexity across all groups and users of {$\mathcal{O}\big(K^2 G (R/G)^2\big)$, and the computation of the \ac{FP} auxiliary variables, which further amounts to $\mathcal{O}\big(K^2N\big)$.}
%
%
Furthermore, the initialization and projection {described in \cite{FangCL2024} have computational complexity} $\mathcal{O}\big(R^3/G^2\big)$, which stems from the \ac{SVD} computation.
Lastly, the tangential projection in \eqref{eq:tangentproj} has a complexity of $\mathcal{O}\big(R^3\big)$.
Considering $I$ iterations of the algorithm, and the case $L \gg 1$, the total complexity of the algorithm can then be expressed as
\vspace{-0.5ex}
\begin{equation}
\label{eq:complexity_simple}
\mathcal{O}\big(IL(KR^3/G^2 + K^2R^2 + K^2RN)\big).
\vspace{-0.5ex}
\end{equation}

Table \ref{tab:complexity} summarizes the simplified computational complexity associated with the different \ac{BD-RIS} architectures, {enabling a direct} comparison with {the \ac{SotA} method from \cite{YahyaOJCS2024}}\footnote{Notice that the {aforementioned \ac{SotA} technique} assumes that the number of \ac{TX} antennas at the \ac{BS} equals the number of single-antenna users ($K=N$). {The proposed method} does not require such assumption and therefore can be considered a generalization of the latter in this regard.}, where $t$ denotes the number of iterations required by {the interference nulling algorithm \cite{YahyaOJCS2024}}.
{In addition, we compare the proposed method against the joint scattering and beamforming matrix design approach based on the \ac{pp-ADMM} framework \cite{ZheyuArxiv2024}. 
Accounting for $T_{\mathrm{ADMM}}$ iterations required to satisfy the convergence criterion, the overall computational complexity of the \ac{pp-ADMM} algorithm scales as
\vspace{-0.5ex}
\begin{equation}
\mathcal{O}\Big(T_{\mathrm{ADMM}}\Big[R^{3}K + \big(\min\!\big\{\sum\nolimits_{i=1}^{R}|S_i|,\,2RK\big\}\big)^{3} \Big]
\Big),
\label{eq:complexity_admm}
\vspace{-0.5ex}
\end{equation}
where $S_i$ denotes the set of non-zero entries ($i.e.$, active connections) in the $i$-th row of the BD-RIS scattering matrix, such that $\sum_{i=1}^{R}|S_i|$ represents the total number of independent optimization variables determined by the connectivity structure of the BD-RIS.}\newpage

\begin{figure}[H]
\centering
\includegraphics[width=1\linewidth]{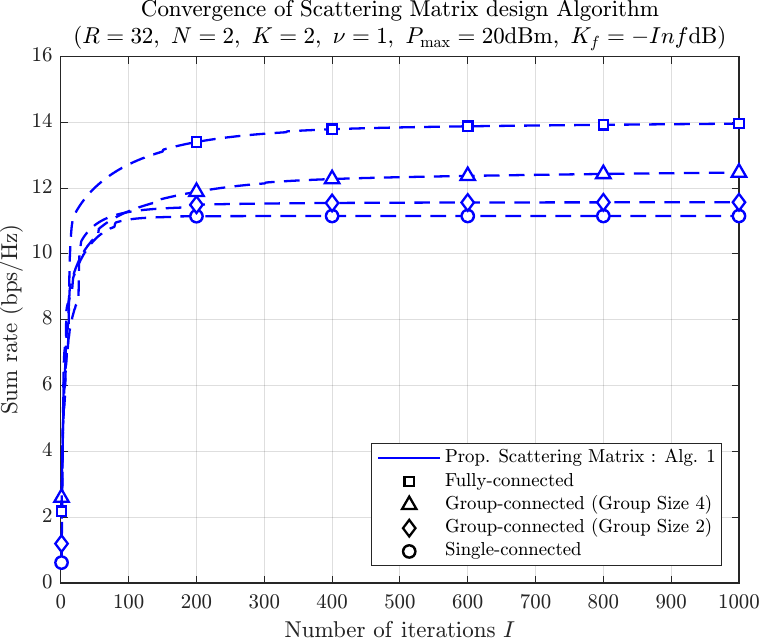}
\vspace{-3ex}
\caption{Convergence of Algorithm~\ref{alg:cga} vs. number of iterations $I$ for different connectivity structures, with $P_{\mathrm{max}} = 20\,\mathrm{dBm}$, {$K = 2$, $N = 2$}, $R = 32$, {and $K$ factor $K_f = -\infty$\,dB}.}
\label{fig:converge}
\vspace{-2ex}
\end{figure}

The convergence behavior of Algorithm~\ref{alg:cga} is shown in Figure~\ref{fig:converge} for the different \ac{BD-RIS} architectures considered, $i.e.$, single-, group-, with group size of 2 and 4, and fully-connected, respectively.
{For clarity,} the figure represents results generated for one Monte-Carlo iteration, hence the reason for the non-smooth convergence curves.

Each \ac{BD-RIS} architecture is represented by a distinct marker, as indicated in the legend of Figure \ref{fig:converge}, and this notation is used consistently throughout the article to facilitate cross-referencing of results.

%
The simulation settings are as follows: convergence tolerance of $\epsilon = 10^{-8}$, maximum number of \ac{CGA} iterations $I = 8000$, maximum Armijo line search steps $L = 200$, sufficient increase coefficient for the stepsize of $2\times 10^{-11}$, initial stepsize $c_{\text{init}} = 1$, contraction factor for the stepsize $c_{\text{dec}} = 0.75$, and the symmetry enforcing penalty constant $\nu = 1$. 
Results indicate that more complex architectures require a larger number of iterations to converge;
for instance, the single-connected case converges in approximately {$100$} iterations, while the full-connected case requires around {$600$} iterations under the same conditions.

Additionally, the convergence speed of the algorithm depends on the number of \acp{RE} and the maximum transmit power $P_{\mathrm{max}}$, as well as on the initial scattering matrix, which in this article is chosen as a random unitary symmetric matrix.
{Regarding solution quality, since the reformulated optimization problem in \eqref{eq:P3a_obj} is convex, the \ac{FP}-based approach applied to the originally non-convex problem is guaranteed to converge to a stationary point, as the adopted approximation yields a monotonically non-decreasing objective sequence \cite{ShenTSP2018_1,ShenTSP2018_2
}.
However, due to the non-convexity of the original problem, the stationary point generally corresponds to a local optimum.}


As validated in Figure \ref{fig:converge}, the number $I$ of iterations required for convergence is considered for the scenario where $P_{\mathrm{max}} = 20\mathrm{dBm}$, {$K = 2$ users, $N = 2$ \ac{TX} antennas, and $R = 32$ \acp{RE}, under Rayleigh fading conditions, $i.e.$, with a $K$-factor $(K_f)$ corresponding to $-\infty$\,dB}.
Specifically, these are {$I_{\mathrm{SC}} = 100$, $I_{\mathrm{GC}_2} = 200$, $I_{\mathrm{GC}_4} = 400$ and $I_{\mathrm{FC}} = 600$,} corresponding to the single-, group-, with group size 2 and 4, and fully-connected architectures, respectively.

As can be learned from Table \ref{tab:complexity} {and \eqref{eq:complexity_admm}}, our proposed algorithm offers scalability, since its computational complexity grows with lower exponents on key system parameters such as the number of users $K$, the number of \ac{RIS} elements $R$ and the number of \ac{BS} antennas $N$.
This advantage is in addition to the superior communications performance, which shall be thoroughly demonstrated in the next section.

For systems of smaller sizes, however, the performance advantage of the proposed method may come at the cost of a higher complexity compared to the \ac{SotA} techniques \cite{YahyaOJCS2024,ZheyuArxiv2024}, due to the large number of iterations required for convergence of the algorithm, as shown in Figure \ref{fig:converge}.
The improved sum-rate and increased robustness observed across multiple \ac{BD-RIS} architectures, resulting from the algorithm's design, are thoroughly presented and analyzed in the subsequent section.

\section{Simulation Results}
\label{sec:simresults}

This section aims to showcase the effectiveness of the proposed \ac{BD-RIS}-aided system in enhancing communication performance via computer simulations executed with the following parameters. 
{Unless otherwise specified, it} will be assumed throughout the {simulations} that the \ac{BS} is equipped with {$N = 2$} \ac{TX} antennas {and serves} {$K = 2$ single \ac{RX} antenna users\footnote{{This choice is motivated by the goal of isolating and evaluating the performance gains introduced by the \ac{BD-RIS}. 
In more heavily loaded scenarios, these gains may be partially masked by additional beamforming gains.}}}, the \ac{RIS} structures consist of $R = 32$ \acp{RE} and the maximum transmit power $P_\mathrm{max}$ ranges from $0$ to $20\,\mathrm{dBm}$.
Furthermore, the parameters utilized in Algorithm \ref{alg:cga} are the same as those previously used to generate the results shown in Figure \ref{fig:converge}.
{In addition to evaluating the proposed scattering matrix design, the numerical results present the achieved sum-rate when the optimized \ac{BD-RIS} is combined with a fixed \ac{MMSE}-based beamformer. 
Joint optimization of the scattering matrix and beamforming is not considered, allowing the impact of the proposed scattering matrix design to be assessed independently.}

\begin{figure}[H]
\centering
\subfloat[Single-connected\label{fig:manopt_SC}]{
\includegraphics[width=.98\columnwidth]{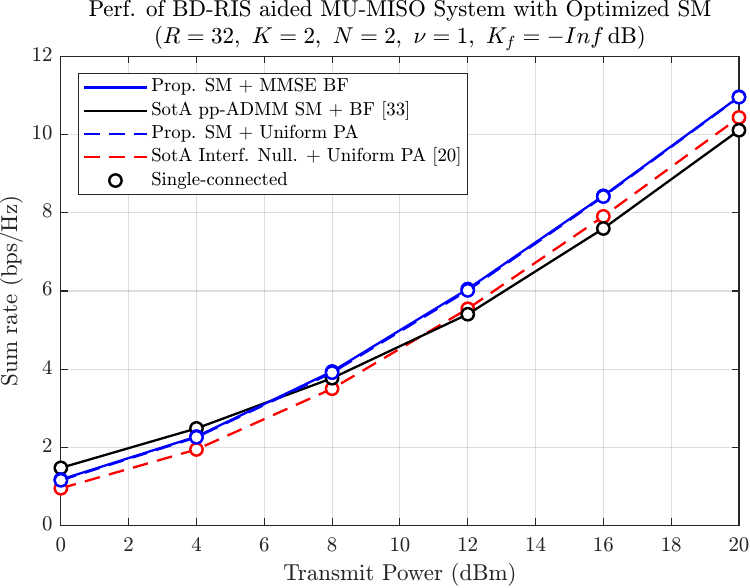}
}  \\[1.5ex]
\subfloat[Group-connected\label{fig:manopt_GC}]{
\includegraphics[width=.98\columnwidth]{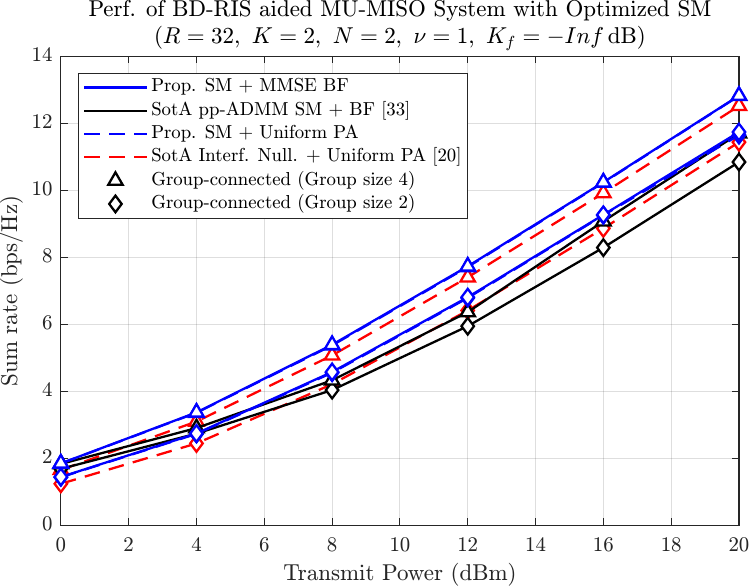}
}  \\[1.5ex]
\subfloat[Fully-connected\label{fig:manopt_FC}]{
\includegraphics[width=.98\columnwidth]{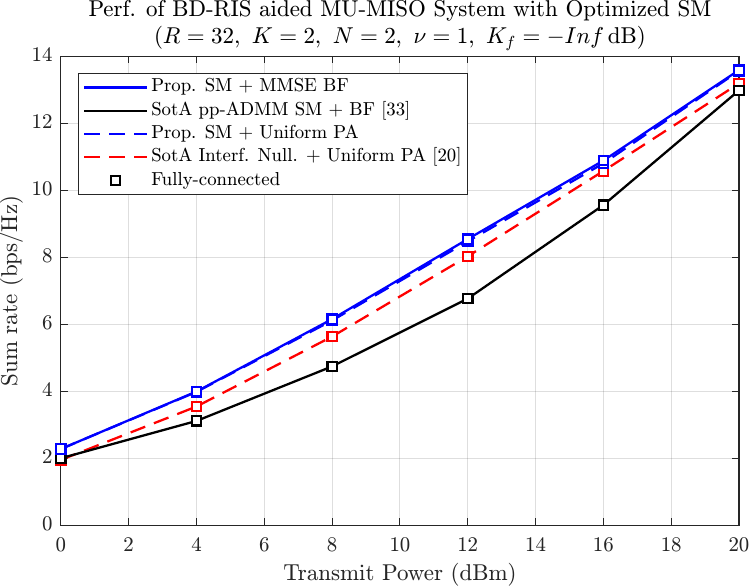}
}
\caption{Comparison of sum-rate performance of the proposed vs. \ac{SotA} \cite{YahyaOJCS2024,ZheyuArxiv2024} scattering matrix design with uniform power allocation, considering the single-connected ``SC", group-connected ``GC", with group sizes of 2 and 4, ``GC(2)" and ``GC(4)", and the fully-connected ``FC" architecture.}
\label{fig:manopt}
\end{figure}

{T}he channel model adopted in this article {considers both} i.i.d. Rayleigh {and Rician small-scale} fading with distance-dependent large-scale {attenuation} expressed as $\Upsilon(d) = C_0(\frac{d}{d_0})^{-\rho}$, where $C_0 = -30\,\mathrm{dB}$ {denotes} the reference channel gain at a reference distance $d_0 = 1$\,m, $\rho = 2.2$ is the path-loss exponent, $d = 50$\,m is the distance between the \ac{BS} and \ac{BD-RIS}, {and $d=2.5$\,m is the distance between the \ac{BD-RIS} and the users\footnote{Although i.i.d fading channels are {not entirely consistent} with \ac{2D} surfaces, as shown in \cite{BjornsonWCL2021}, {this model} is adopted here {only to enable} a fair comparison with the \ac{SotA} methods \cite{YahyaOJCS2024, ZheyuArxiv2024}, since the scattering matrix design itself does not require any specific model to be feasible. 
{The fading parameters and large-scale attenuation settings are selected consistently with the aforementioned \ac{SotA} methods to ensure meaningful performance comparisons. 
Importantly, the proposed scattering matrix design framework does not rely on any specific small-scale fading assumption.}}}.
{Specifically, some numerical results assume Rician fading with $K_f = 2\,\mathrm{dB}$, while others consider the Rayleigh case corresponding to $K_f=-\infty\,\mathrm{dB}$.}
{Lastly, t}he system is assumed to operate at $2.4$GHz, and each user experiences identical noise with a noise power $N_0 = -80\mathrm{dBm}$.

The first set of results is presented in Figure \ref{fig:manopt} to demonstrate the effectiveness of the proposed method in enhancing system performance.
Different architectures, such as the single-, group-, and fully-connected are considered, and results comparing the sum-rate performance achieved by the proposed and the \ac{SotA} {techniques} \cite{YahyaOJCS2024,ZheyuArxiv2024} are provided.
Notably, the results indicate that the proposed scattering matrix ``SM'', {combined with either uniform power allocation ``PA'', or a fixed MMSE-based beamformer ``BF'',} outperforms both \ac{SotA} {counterparts considered.}

{In particular, the interference nulling SM design with uniform PA \cite{YahyaOJCS2024} and the joint scattering and beamforming matrix design based on the \ac{pp-ADMM} framework \cite{ZheyuArxiv2024}, are outperformed with all group architectures, under Rayleigh fading conditions with $K_f = -\infty$\,dB and for a small number of users and \ac{TX} antennas at the \ac{BS}, i.e., $K = 2$ and $N = 2$.}
%
\vspace{-3ex}
\begin{figure}[H]
\centering
\includegraphics[width=0.98\linewidth]{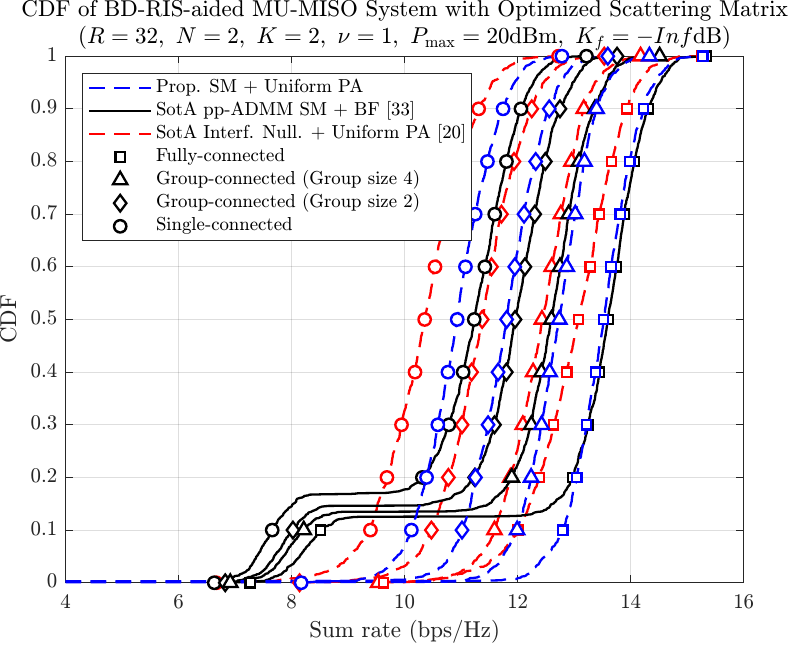}
\caption{CDF of sum-rate performance of the proposed vs. SotA \cite{YahyaOJCS2024,ZheyuArxiv2024} scattering matrix design with uniform power allocation, considering the fully-connected ``FC", group-connected ``GC", with group sizes of 2 and 4, ``GC(2)" and ``GC(4)", and the single-connected "SC" architecture.}
\label{fig:cdf}
\vspace{-2ex}
\vspace{-3ex}
\end{figure}

\begin{figure}[H]
\centering
\subfloat[Single-connected\label{fig:SM_BF_SC}]{
\includegraphics[width=.98\columnwidth]{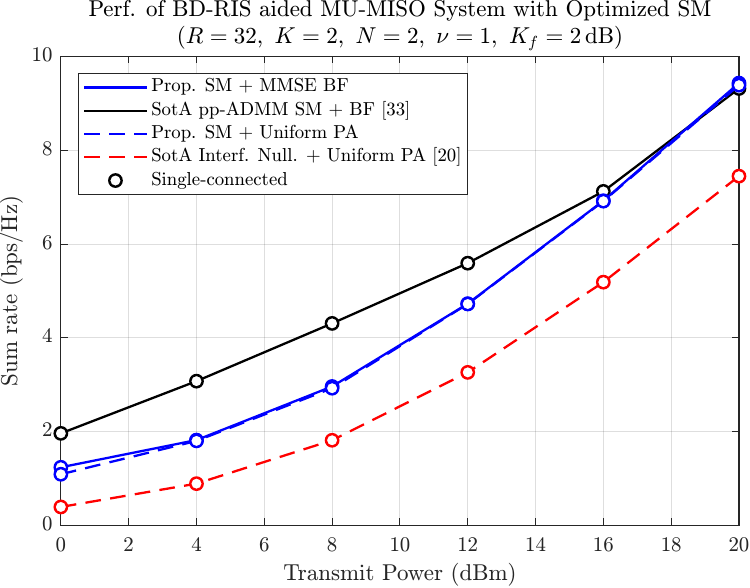}
} \\[1.5ex]
\subfloat[Group-connected\label{fig:SM_BF_GC}]{
\includegraphics[width=.98\columnwidth]{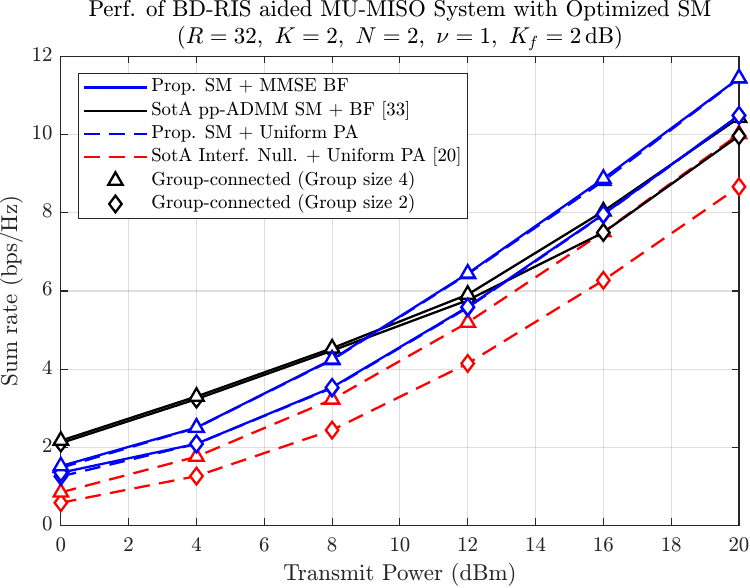}
} \\[1.5ex]
\subfloat[Fully-connected\label{fig:SM_BF_FC}]{
\includegraphics[width=.98\columnwidth]{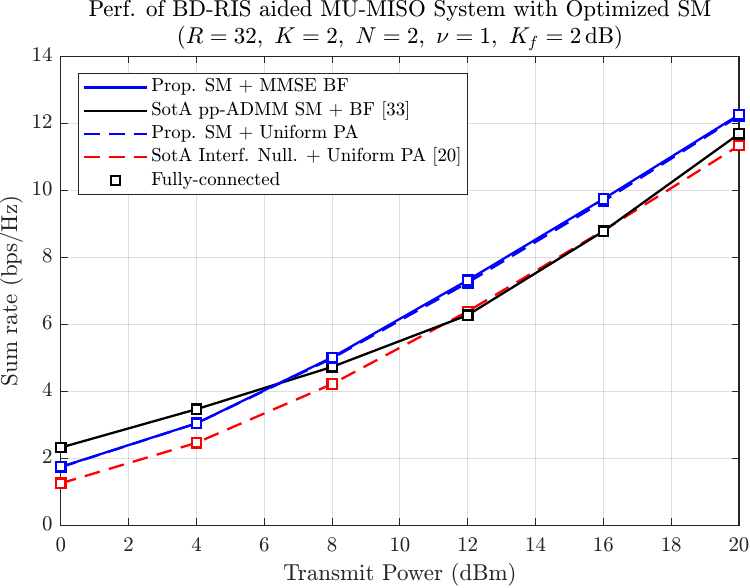}
}
\caption{Comparison of sum-rate performance using different beamforming schemes, considering the proposed and SotA \cite{YahyaOJCS2024,ZheyuArxiv2024} single-connected ``SC", group-connected ``GC", with group sizes of 2 and 4, ``GC(2)" and ``GC(4)", and the fully-connected ``FC" architecture.}
\label{fig:SM_BF}
\end{figure}

Complementing the assessment in terms of sum-rate performance, Figure \ref{fig:cdf} further shows the \ac{CDF} of the sum-rate performance of the proposed \ac{BD-RIS} scattering matrix design, alongside the \ac{SotA} methods \cite{YahyaOJCS2024,ZheyuArxiv2024} for all considered architectures, $i.e.$, single-, group- with group sizes of 2 and 4, and fully-connected. 

{More specifically, the group-connected architecture with group size 4, when using the proposed method, clearly outperforms the considered \ac{SotA} designs with the same architecture. The fully-connected architecture with the proposed method achieves, on average, slightly better performance than \cite{ZheyuArxiv2024}. Architectures with lower complexity, namely the single-connected architecture and the group-connected architecture with group size 2, achieve slightly lower performance than \cite{ZheyuArxiv2024}; however, all proposed architectures still outperform the \ac{SotA} interference nulling design in \cite{YahyaOJCS2024}.}

Next, Figure \ref{fig:SM_BF} presents a comparison of the sum-rate performance of the proposed scattering matrix ``SM'', {in combination with either uniform power allocation ``PA'', or a fixed MMSE-based beamformer ``BF'', against the corresponding \ac{SotA} schemes, namely the interference nulling SM design with uniform PA \cite{YahyaOJCS2024} and the joint scattering and beamforming matrix design based on the \ac{pp-ADMM} framework \cite{ZheyuArxiv2024}.

The comparison is conducted across all considered architectures under Ricean fading conditions with $K_f = 2$\,dB, and for a small number of users and \ac{TX} antennas at the \ac{BS}, specifically $K = 2$ and $N = 2$, highlighting the robustness of the proposed method across different fading conditions.}
{We remark that notable performance gains are particularly observed for the group-connected architecture, with group size 2, where the proposed method outperforms both \ac{SotA} methods, even those with a higher connectivity, by an order of 2, in the high transmit power regime.}

\begin{figure}[H]
\centering
\includegraphics[width=0.98\linewidth]{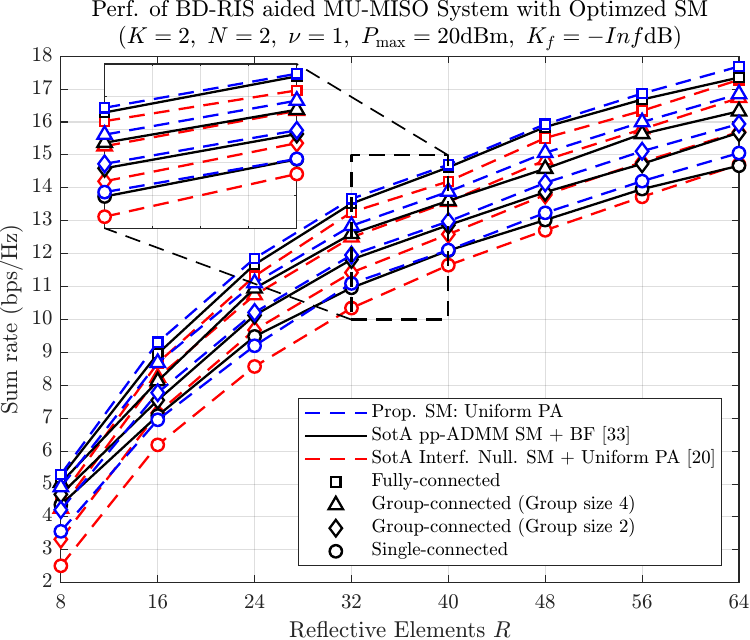}
\caption{Comparison of sum-rate performance of the proposed vs. \ac{SotA} \cite{YahyaOJCS2024 ,ZheyuArxiv2024} {methods} as a function of the number of \acp{RE}, i.e. $R$, considering the fully-connected ``FC", group-connected ``GC", with group sizes of 2 and 4, ``GC(2)" and ``GC(4)", and the single-connected "SC" architecture.}
\label{fig:sRvsN}
\vspace{-1ex}
\end{figure}

Finally, Figure \ref{fig:sRvsN} presents the sum-rate performance of the proposed method, {in combination with uniform power allocation ``PA",} in comparison to the {\ac{SotA} methods \cite{YahyaOJCS2024,ZheyuArxiv2024}}, as a function of the number of \acp{RE} $R$.
This result further builds upon the claim that for all \ac{BD-RIS} architectures considered, $i.e.$, single-, group-, and fully-connected, and under {the Rayleigh fading} scenario, the proposed scattering matrix consistently outperforms the \ac{SotA} benchmarks \cite{YahyaOJCS2024,ZheyuArxiv2024}.
Importantly, the performance gain is maintained as the number of reflective elements increases, which highlights the robustness of the proposed approach, and showcases that the proposed solution is scalable.

\vspace{-1ex}
\section{Conclusion}
\label{sec:conc}

A novel sum-rate maximization scheme for \ac{RBD-RIS}-aided \ac{MU}-\ac{MISO} systems is presented.
To this end, an optimization problem was formulated to determine the \ac{BD-RIS} scattering matrix configuration that maximizes the users' sum-rate, and solved using a Riemannian manifold steepest descent method with a closed-form gradient.
Simulation results validated the effectiveness of the proposed approach, showing consistent performance gains over the \ac{SotA}, where the improvements arise from jointly exploiting the geometric structure of the manifold and the direct sum-rate based problem formulation.  
{Additionally, the simulation results were compared against joint scattering matrix and beamforming designs. Under the considered scenarios, the proposed method outperformed the design with fixed beamforming, which opens up the possibility for beamforming optimization.

This is promising, as further beamforming design optimization is expected to provide even greater performance improvements.}
Overall, the results show that manifold-based scattering matrix optimization enables scalable and robust performance in \ac{RIS}-aided systems. 
Future work will extend the framework to include transmit beamforming design.
%



\appendices
\label{sec:appendix}

\vspace{-2ex}
\section{Proof of Lemma \ref{lem:gradAconjA}}
\label{app:lemma1}

\begin{lemma}
\label{lem:gradAconjA}
Consider $\mathbf{A} \in \mathbb{C}^{N \times N}$ with entries defined as
\vspace{-1ex}
\begin{equation}
\mathbf{A} = {\begin{pmatrix}
a_{1,1}  & \cdots & a_{1,N} \\[-1ex]
\vdots   & \ddots & \vdots  \\[-1ex]
a_{N,1}  & \cdots & a_{N,N}
\end{pmatrix}}.
\vspace{-1ex}
\end{equation}

Then
\vspace{-1ex}
\begin{equation}
\nabla_{\mathbf{A}}\mathrm{Tr}\{\mathbf{A}^{\star}\mathbf{A}\} = 2\frac{\partial\mathrm{Tr}\{\mathbf{A}^{\star}\mathbf{A}\}}{\partial\mathbf{A}^{\star}} = 2\mathbf{A}^T,
\vspace{-1ex}
\end{equation}
where
\begin{equation}
\mathrm{Tr}\{\mathbf{A}^{\star}\mathbf{A}\} = \sum_{i=1}^N\sum_{k=1}^N a_{i,k}^{\star}a_{k,i},
\end{equation}
with $a_{i,k}^{\star}$ and $a_{k,i} \in \mathbb{C}$, such that they can be written in the form $x=\Re{x}+j\Im{x}$.
\end{lemma}

\vspace{-1ex}
\emph{Proof:} For simplicity, let $\phi = \mathrm{Tr}\{\mathbf{A}^\star \mathbf{A}\}$.
Then,
\vspace{-1ex}
\begin{equation}
\frac{\partial \phi}{\partial \Re{\mathbf{A}}}  = {\begin{pmatrix}
\frac{\partial \phi}{\partial \Re a_{1,1}} & \cdots & \frac{\partial \phi}{\partial \Re a_{1,N}} \\[-1ex]
\vdots & \ddots & \vdots \\
\frac{\partial \phi}{\partial \Re a_{N,1}} & \cdots & \frac{\partial \phi}{\partial \Re a_{N,N}}
\end{pmatrix} =2\Re\mathbf{A}^T,}
\vspace{-1ex}
\end{equation}
and
\vspace{-1ex}
\begin{equation}
j\frac{\partial \phi}{\partial \Im{\mathbf{A}}}  = {\begin{pmatrix}
\frac{\partial \phi}{\partial \Im a_{1,1}}  & \cdots & \frac{\partial \phi}{\partial \Im a_{1,N}} \\[-1ex]
\vdots  & \ddots & \vdots \\
\frac{\partial \phi}{\partial \Im a_{N,1}} & \cdots & \frac{\partial \phi}{\partial \Im a_{N,N}}
\end{pmatrix} =2j\Im\mathbf{A}^T.}
\vspace{-1ex}
\end{equation}
\newpage

Straightforwardly, the Generalized Complex Derivative and Conjugate Complex Derivative {\cite[eq.(229), eq.(230)]{Petersen2008}} are obtained as
\vspace{-1ex}
\begin{eqnarray}
&\frac{\mathrm{d}\phi}{\mathrm{d}\mathbf{A}} = \frac{1}{2} \Big(2\Re\mathbf{A}^T-2j\Im\mathbf{A}^T\Big) = \mathbf{A}^H,&\\
&\frac{\mathrm{d}\phi}{\mathrm{d}\mathbf{A}^{\star}} = \frac{1}{2} \Big(2\Re\mathbf{A}^T+2j\Im\mathbf{A}^T\Big) = \mathbf{A}^T.&
\vspace{-1ex}
\end{eqnarray}

As evidence that $\phi$ is a real-valued function, the following is provided 
\vspace{-2ex}
\begin{equation}
\mathrm{Tr}\{\mathbf{A}^{\star}\mathbf{A}\} = \sum_i^N|a_{i,i}|^2 + 2\Re\Big\{\sum_i^N\sum_{k=i+1}^Na_{i,k}^{\star}a_{k,i}\Big\},
\vspace{-1ex}
\end{equation}
as such, the Complex Gradient Matrix {\cite[eq.(233)]{Petersen2008}} can be computed as follows
\vspace{-1ex}
\begin{equation}
\nabla_{\mathbf{A}} \mathrm{Tr}\{\mathbf{A}^{\star}\mathbf{A}\} =  2\frac{\partial\mathrm{Tr}\{\mathbf{A}^{\star}\mathbf{A}\}}{\partial\mathbf{A}^{\star}} = 2\mathbf{A}^T,
\vspace{-1ex}
\end{equation}
hence, \eqref{eq:Transpose} is proven. \hfill$\square$

\vspace{-3ex}
\section{Proof of Lemma \ref{lem:lem2}}
\label{app:lemma2}

\begin{lemma}
\label{lem:lem2}

Consider $\mathbf{a} \in \mathbb{C}^{K\times 1}$, $\mathbf{B} \in \mathbb{C}^{K\times K}$, and $\mathbf{c}\in\mathbb{C}^{K\times 1}$, such that $d = \mathbf{a}^T\mathbf{B}\mathbf{c}$, where $ d \in \mathbb{C}$ and
\vspace{-1ex}
\begin{equation}
d ={
\left(\! \begin{array}{c}
a_1\\[-1ex]
\vdots \\[-1ex]
a_N
\end{array}\! \right)^T \!\!\!
\left( \!\begin{array}{cccc}
B_{1,1}  & \cdots & B_{1,N} \\[-1ex]
\vdots  & \ddots & \vdots \\
B_{N,1}  & \cdots & B_{N,N}
\end{array} \!\right) \!\!\!
\left( \!\begin{array}{c}
c_1 \\[-1ex]
\vdots \\[-1ex]
c_N
\end{array} \!\right).}
\vspace{-1ex}
\end{equation}

Consider the function $f(d)=|d|^2$. The gradient of $f(d)$ with respect to $\mathbf{B}$ is given as
\vspace{-1ex}
\begin{equation}
\nabla_{\mathbf{B}}|d|^2 = d\mathbf{a}^{\star}\mathbf{c}^{H},
\vspace{-1ex}
\end{equation}
where $\mathbf{a} = \Re\mathbf{a} + j\Im\mathbf{a}$, $\mathbf{B} = \Re\mathbf{B} + j\Im\mathbf{B}$ and $\mathbf{c} = \Re\mathbf{c} + j\Im\mathbf{c}$, where $\Re \backslash \Im \mathbf{a}$, $\Re \backslash \Im \mathbf{c} \in \mathbb{R}^{K \times 1}$ and $\Re \backslash \Im \mathbf{B} \in \mathbb{R}^{K \times K}$.
\end{lemma}
\begin{proof}

Following the steps for calculating the gradient of complex valued functions, as described in {\cite[eq.(229), eq.(230), eq.(233)]{Petersen2008}}.
The derivatives with respect to the real and imaginary parts of $\mathbf{B}$ are calculated, as shown below
\vspace{-1ex}
\begin{equation}
\frac{\partial |d|^2}{\partial \Re \mathbf{B}} = \frac{\partial |d|^2}{\partial d} \frac{\partial d}{\partial \Re \mathbf{B}}  + \frac{\partial |d|^2}{\partial d^{\star}}\frac{\partial d^{\star}}{\partial \Re \mathbf{B}},
\vspace{-1ex}
\end{equation}
and
\vspace{-1ex}
\begin{equation}
\frac{\partial |d|^2}{\partial \Im \mathbf{B}} = \frac{\partial |d|^2}{\partial d} \frac{\partial d}{\partial \Im \mathbf{B}}  + \frac{\partial |d|^2}{\partial d^{\star}}\frac{\partial d^{\star}}{\partial \Im \mathbf{B}}.
\vspace{-1ex}
\end{equation}

Straightforwardly
{
\vspace{-1ex}
\begin{equation}
\frac{\partial |d|^2}{\partial d}  = d^{\star} \quad \text{and} \quad
\frac{\partial |d|^2}{\partial d^{\star}}  = d.
\end{equation}}
\vspace{-2ex}

Expanding $ d = \mathbf{a}^T\mathbf{B}\mathbf{c}$ and $d^{\star} = \mathbf{c}^{H}\mathbf{B}^{H}\mathbf{a}^{\star}$ as a product of complex numbers. The following derivatives are obtained
\vspace{-1ex}
\begin{eqnarray}
&\frac{\partial d}{\partial \Re \mathbf{B}}  = 
{
\begin{pmatrix}
\frac{\partial d}{\partial \Re B_{1,1}}  & \cdots & \frac{\partial d}{\partial \Re B_{1,N}} \\[-1ex]
\vdots & \ddots & \vdots \\
\frac{\partial d}{\partial \Re B_{N,1}}  & \cdots & \frac{\partial d}{\partial \Re B_{N,N}}
\end{pmatrix}
= \mathbf{a}\mathbf{c}^T,}\\
&\frac{\partial d}{\partial \Im \mathbf{B}}  = 
{
\begin{pmatrix}
\frac{\partial d}{\partial \Im B_{1,1}}  & \cdots & \frac{\partial d}{\partial \Im B_{1,N}} \\[-1ex]
\vdots & \ddots & \vdots \\
\frac{\partial d}{\partial \Im B_{N,1}}  & \cdots & \frac{\partial d}{\partial \Im B_{N,N}}
\end{pmatrix}
= j\mathbf{a}\mathbf{c}^T.}&
\end{eqnarray}

Similarly, for $d^{\star}$
\begin{equation}
\color{black}
\frac{\partial d^{\star}}{\partial \Re \mathbf{B}} = \mathbf{a}^{\star}\mathbf{c}^H \quad \text{and} \quad
\frac{\partial d^{\star}}{\partial \Im \mathbf{B}} = -j\mathbf{a}^{\star}\mathbf{c}^H.
\end{equation}

Finally {we obtain}
\begin{eqnarray}
&\frac{\partial |d|^2}{\partial \Re \mathbf{B}} = d^{\star}\mathbf{a}\mathbf{c}^T + d\mathbf{a}^{\star}\mathbf{c}^H,&\\
&\frac{\partial |d|^2}{\partial \Im \mathbf{B}} = jd^{\star}\mathbf{a}\mathbf{c}^T -jd\mathbf{a}^{\star}\mathbf{c}^H,&\\
&\frac{\mathrm{d}|d|^2}{\mathrm{d}\mathbf{B}^{\star}} = \frac{1}{2}\Big(\frac{\partial |d|^2}{\partial \Re \mathbf{B}}+ j\frac{\partial |d|^2}{\partial \Im \mathbf{B}}\Big) = d\mathbf{a}^{\star}\mathbf{c}^H.&
\end{eqnarray}

As the absolute value squared of a complex number is a real value, $i.e.$, $|d|^2$ is a real function, the gradient of $f(d)$ with respect to $\mathbf{B}$ is given as
\vspace{-1ex}
\begin{eqnarray}
\nabla_{\mathbf{B}} |d|^2 = 2\frac{\mathrm{d}|d|^2}{\mathrm{d}\mathbf{B}^{\star}} = 2 d\mathbf{a}^{\star}\mathbf{c}^H = 2 \big(d^{\star}\mathbf{a}\mathbf{c}^T\big)^{\star},
\vspace{-1ex}
\end{eqnarray}
hence, {\eqref{eq:term1_grad} and~\eqref{eq:term2_grad}} are proven.
\end{proof}


\vspace{-3ex}

\end{document}